\documentclass[runningheads]{llncs}
\usepackage[T1]{fontenc}
\usepackage{graphicx}
\usepackage[dvipsnames]{xcolor}
\usepackage{amsmath,amssymb,amsfonts}%
\usepackage{xargs}
\usepackage{mathtools}
\usepackage{turnstile}
\usepackage{stmaryrd}
\usepackage{listings}
\usepackage{proof}
\usepackage{xfrac}
\usepackage{paralist}
\usepackage{thmtools}
\usepackage{algorithm}
\usepackage{algpseudocodex}
\usepackage{atlas-paper}
\usepackage{code}

\usepackage{booktabs}
\usepackage{tabularx}
\usepackage{threeparttable}
\usepackage{multirow}
\usepackage{makecell}
\usepackage{subcaption}
\usepackage{mathpartir}
\usepackage{mdframed}
\usepackage[shortlabels,inline]{enumitem}

\usepackage[backend=bibtex,bibstyle=lncs]{biblatex}
\usepackage[disable]{todonotes}

\newif\ifextended
\extendedtrue

\begin{document}
\title{Automated Amortised Analysis \\ of Skew Heaps and Leftist Heaps}

\todo{add double affiliation for Florian}
\author{Armin Walch\inst{1}\orcidID{0000-0002-2298-9353} \and
Georg Moser\inst{1}\orcidID{0000-0001-9240-6128} \and
Berry Schoenmakers\inst{2}\orcidID{0000-0001-6273-8930} \and
Florian Zuleger\inst{3,4}\orcidID{0000-0003-1468-8398}}
\authorrunning{A. Walch et al.}
%
\institute{
  Department of Computer Science, University of Innsbruck, Austria\\
  \email{\{armin.walch,georg.moser\}@uibk.ac.at} \and
  Department of Mathematics and Computer Science, Eindhoven University of Technology, Netherlands\\
  \email{berry@win.tue.nl}
  \and
  Institute of Logic and Computation, Vienna University of Technology, Austria\\
  \email{florian.zuleger@tuwien.ac.at}
  \and
  TUM School of Computation, Information and Technology, TU München, Germany
}

\maketitle              
\begin{abstract}
We study the fully automated amortised analysis of purely functional data structures like \emph{skew heaps}, as well as \emph{weight}- and \emph{rank-biased} leftist heaps. For that we generalise earlier works on automated amortised resource analysis by developing a type inference based approach with a generic type system. This allows for modular reasoning and the inference of precise and optimal cost bounds.

More specifically, we extend the work on the ATLAS system by Leutgeb et al.\, which was developed to cover the analysis of splay trees and some closely related data structures. To enable the analysis of skew heaps, however, and the even more challenging (amortised) analysis of leftist heaps, we have developed a range of new techniques for type-based automated analysis. By introducing a generic type system we allow for arbitrary (classes of) potential functions, compared to the use of hard-coded potential functions in ATLAS, which we have implemented in Haskell in an entirely modular way. We have also greatly enhanced the existing type inference algorithm by extensions in multiple directions, including path-sensitive reasoning, data structure invariants, and template parameters for piecewise defined potential functions.
We show how our newly developed system supports the use of all known potential functions for analysing skew heaps and leftist heaps, confirming the known bounds.
\end{abstract}

\section{Introduction}

\emph{Amortised analysis}~\cite{sleator1985self,tarjan1985amortized} determines worst-case runtime bounds by averaging costs over sequences of operations rather than analysing each operation in isolation. Using the potential method, due to Tarjan and Sleator, one assigns a \emph{potential function} to a data structure whose stored potential can “pay for’’ occasional expensive operations. This perspective has been particularly influential in the design of \emph{self-adjusting} data structures such as splay trees~\cite{sleator1985self,tarjan1985amortized}, which avoid maintaining structural invariants while still achieving good amortised performance.

\emph{Automated amortised resource analysis} (AARA), originating with Hofmann and Jost~\cite{hofmann2003static}, aims to automate such proofs by inferring resource-annotated types. Each typing rule generates constraints, which are solved by an SMT solver to obtain amortised costs; the resulting type signature yields the resource bound.

Verifying textbook amortised data structures remains a key benchmark for modern analysis tools. Nipkow and Brinkop~\cite{nipkow2019amortized} provided formal proofs of classical examples, and Leutgeb et al.~\cite{leutgeb2021atlas,leutgeb2022automated} introduced the \atlas\ prototype, which automatically analyses \emph{splay trees}, \emph{splay heaps}, and \emph{pairing heaps}. However, \atlas\ does not support \emph{skew heaps} or \emph{leftist heap variants}---structures that present new challenges and, in the case of leftist heaps, admit strong amortised bounds, as recently shown by Schoenmakers~\cite{schoenmakers2024amortized}. That work also posed the open problem of establishing a better bound for rank-biased heaps and explicitly pointed to innovations such as \atlas\ as possible enablers.

We found that the original \atlas\ analysis is insufficient for these structures. It cannot express the piecewise potential used in the classical skew-heap analysis~\cite{sleator1986self}, nor the sum-of-logs potential of Kaldewaij and Schoenmakers~\cite{kaldewaij1991derivation}, which requires negative terms. Moreover, \atlas\ lacks support for data-structure invariants and path-sensitive reasoning, both essential for analysing weight-biased and rank-biased leftist heaps.

To overcome these limitations, we introduce several extensions to the original analysis. First, we add new potential functions with Iverson terms, allowing us to model the original piecewise analysis. Second, we generalise \atlas’s sum-of-logs potential function. Third, for leftist heaps, we introduce data-structure invariants and enable path sensitivity. Rather than relying on ad hoc extensions, we develop a more abstract type system that cleanly separates the core reasoning from potential-function families; we realise this design in a modular Haskell prototype that decouples these families from the type system, enabling the support of fundamentally different classes of potential functions. The new type system applies to arbitrary data structures and constitutes a contribution in its own right; we instantiate it for leftist heaps and skew heaps as particularly challenging benchmarks.

\definecolor{tabbg}{gray}{0.94}
\begin{figure}[t]

  \begin{subfigure}{\textwidth}
    \renewcommand{\arraystretch}{1.1}
    \setlength{\fboxsep}{0pt}
    \colorbox{tabbg}{
      \hspace{-2.1mm}
      \begin{minipage}{0.65\textwidth}
        \begin{tabular}{lr}
          \toprule
          $\Phi$ & $\Phi(\tree{t}{a}{u})$\\
          \midrule
          $\Phi_{[t < u]}$ & $\Phi_{[t < u]}(t) + [t < u] + \Phi_{[t < u]}(u)$\\
          $\Phi_{u}$ & $\Phi_{u}(t) + \sfrac{1}{2}\log\abs{u} + \Phi_{u}(u)$\\
          $\Phi_{-t}$ &$\Phi_{-t}(t) + \sfrac{1}{2}\log(\abs{t} + \abs{u}) - \sfrac{1}{2}\log\abs{t} + \Phi_{-t}(u)$\\
          $\Phi_{\phi}$ & $\Phi_{\phi}(t) + \frac{105}{163}\log(\abs{t} + \abs{u}) - \frac{105}{163}\log\abs{t} + \Phi_{\phi}(u)$\\
          \bottomrule
        \end{tabular}
      \end{minipage}
      \hspace{9.2mm}
      \begin{minipage}{.252\textwidth}
        \renewcommand{\arraystretch}{1.49}
        \begin{tabular}{ll}
          \toprule
          $[\ ]$ & Iverson brackets\\
          $\curlyvee$ &path-sensitivity\\
          $\Box$ & invariant\\
          $^{\ddag}$ & worst case bounds\\
          \bottomrule
        \end{tabular}
      \end{minipage}}
    \caption{Legend}
  \end{subfigure}%
  \medskip

  \begin{subfigure}{\textwidth}
    \begin{tabularx}{\textwidth}{l>{\centering}p{3.5cm}>{\centering}p{5mm}>{\centering}p{5mm}>{\centering}p{5mm}Xlr}
      \toprule
      $\Phi$ & Source & \multicolumn{3}{r}{Type System} && \multicolumn{2}{c}{Inferred Bounds}\\
          & &$[\ ]$ & $\curlyvee$ & $\Box$ && \texttt{meld} & \texttt{del\_min} \\
      \midrule
      \multicolumn{8}{c}{skew heap} \\
      \addlinespace[0.4em]
      $\Phi_{[t < u]}$ &~\cite{sleator1986self,nipkow2019amortized} &  $\checkmark$& $\times$&$\times$ && $3\log(\abs{x} + \abs{y})$ &  $3\log\abs{x}$ \\
    $\Phi_{u}$ &  &$\times$ &$\times$ &$\times$ && $2\log(\abs{x} + \abs{y})$& $\sfrac{3}{2}\log\abs{x}$ \\
    $\Phi_{-t}$ & ~\cite{kaldewaij1991derivation}&  $\times$& $\times$& $\times$&& $\sfrac{3}{2}\log(\abs{x} + \abs{y})$ & $\sfrac{3}{2}\log\abs{x}$\\
    $\Phi_{\phi}$ & ~\cite{kaldewaij1991derivation}&  $\times$& $\times$& $\times$&& $1.44\log(\abs{x} + \abs{y})$ & $1.44\log\abs{x}$\\
    \midrule
      \multicolumn{8}{c}{weight-biased leftist heap}\\
      \addlinespace[0.4em]
    n.a.&~\cite{cho1998weight} & $\times$  & $\checkmark$ & $\checkmark$&& $2\log(\abs{x} + \abs{y})^\ddag$ & $2\log\abs{x}^\ddag$ \\
    $\Phi_{\phi}$ &~\cite{schoenmakers2024amortized} & $\times$& $\checkmark$ & $\times$&&$1.44\log(\abs{x} + \abs{y})$ & $1.44\log\abs{x}$ \\
    \midrule
     \multicolumn{8}{c}{rank-biased leftist heap}\\
      \addlinespace[0.4em]
    n.a.& ~\cite{crane1972linear} & $\times$&  $\checkmark$ & $\checkmark$ && $2\log(\abs{x} + \abs{y})^\ddag$ & $2\log\abs{x}^\ddag$ \\
      $\dag x$ & ~\cite{schoenmakers2024amortized}& $\times$& $\times$ & $\times$&&$\log(\abs{x} + \abs{y})$ & $2\log\abs{x}$ \\
    \bottomrule
    \end{tabularx}%
    \caption{Bounds}
  \end{subfigure}
  \caption{Amortised cost of \texttt{meld} and \texttt{delete\_min} operations for leftist heaps, for different potential functions $\Phi$, inferred by our prototype implementation.}
\label{fig:results}
\end{figure}
Figure~\ref{fig:results} summarises our results, highlighting the type-system features required for the different data structures and potential functions. The bounds match those reported in the literature (with references given in the table), but—crucially—are obtained fully automatically, replacing previously tedious pen-and-paper analyses and mechanised proofs.
For ease of comparison, Figure~\ref{fig:results} reports clean upper bounds in the style of prior work, while the exact bounds computed by our prototype are listed in \ifextended Appendix~\ref{sec:exact-bounds}. \fi
\begin{itemize}
\item We have developed a generic type system for amortised cost analysis that on the one hand can be instantiated to obtain the previous \atlas\ tool and on the other hand easily generalizes to new classes of potential functions. 
\item We extend the existing type inference algorithm for amortised
  analysis by 
  \begin{enumerate*}[(i)]
  \item path-sensitive reasoning;
  \item a form of refinement types, to incorporate data structure
    invariants;
  \item template parameters, which are analogous to ghost variables in Hoare logic. 
  \end{enumerate*}
\item Building on these extensions, we present the first automated amortised analysis of skew heaps, as well as weight- and rank-biased leftist heaps, achieving the best known bounds—optimal for skew and weight-biased heaps~\cite{schoenmakers1997tight, schoenmakers2024amortized}.
\item The analysis has been implemented in a novel Haskell prototype that generalises $\atlas$ and relies on SMT solving via Z3. The prototype implementation will be submitted to the artefact evaluation.
\end{itemize}

\paragraph{Outline.} The remainder of this paper is structured as follows. Section~\ref{sec:overview} provides a high-level overview of skew heaps, leftist heaps, and our technique. Sections~\ref{sec:semantics} and~\ref{sec:aara} present our core language, cost semantics, and type rules, including the main soundness result. Section~\ref{sec:inference} describes our type inference algorithm, and Section~\ref{sec:case-studies} instantiates the analysis for skew heaps and leftist heaps. Related work is discussed in Section~\ref{sec:related-work}, and Section~\ref{sec:conclusion} concludes the paper.


\section{Overview}
\label{sec:overview}
\subsection{Skew Heaps and Leftist Heaps}
We analyse three different meldable heap data structures. Meldable heaps are binary-heap data structures used to implement priority queues, where operation \texttt{meld} merges two heaps while preserving the heap property~\cite{okasaki1998purely, cormen2022introduction}. It forms the core of the primary operations \texttt{insert} and \texttt{delete\_min}, whose complexity therefore depends on an efficient implementation of \texttt{meld}. Many variants of meldable heaps aim to keep the merge path short to limit the cost of this operation~\cite{okasaki1998purely, cormen2022introduction}.

\paragraph{Rank-biased leftist heaps} introduced by Crane in 1972~\cite{crane1972linear} and further developed by Knuth~\cite{knuth1998art}, maintain, for every node in the tree, the invariant that the rank of the left subtree, i.e., the length of its rightmost path, dominates the rank of the right subtree. This results in the left subtree usually being larger than the right subtree, giving the data structure its name.

\paragraph{Weight-biased leftist heaps} utilise the number of nodes of a tree---its weight---as an alternative bias. They were introduced by Cho and Sahni~\cite{cho1998weight}. Both variants merge along the rightmost path and, if the corresponding invariant is violated afterwards, swap the result with the left subtree. They also share the same worst-case cost for \texttt{meld}, namely $2\log(\abs{x} + \abs{y})$ for heaps $x$ and $y$.

\paragraph{Skew heaps} were presented by Sleator and Tarjan as the self-adjusting counter-part of leftist heaps~\cite{sleator1986self}. They support \texttt{meld} with amortised logarithmic complexity in the size of the input heaps, while requiring no invariants, which greatly simplifies their implementation. More specifically, the exact bound for \texttt{meld} has been shown to be $\oldlog_{\phi}(\abs{x} + \abs{y})$, where $\phi$ denotes the golden ratio~\cite{kaldewaij1991derivation}. More recently, Schoenmakers demonstrated that the same amortised upper bound also holds for weight-biased leftist heaps~\cite{schoenmakers2024amortized}.

\paragraph{Implementation.} Figure~\ref{fig:leftist-code} shows a purely functional implementation of a generalised leftist heap in our first-order core language, which is introduced formally in Section~\ref{sec:semantics}. The variants differ only in their balancing strategy, i.e., the implementation of \texttt{bal}. We do not store the measures $\#$ and $\dag$ as tree fields, as this would require proving correctness of the bookkeeping, which lies outside the scope of this work; however it could be achieved semi-automatically, using liquid types~\cite{rondon2008liquid,vazou2013abstract}. Instead, we represent them via the built-in functions \texttt{weight} and \texttt{rank}, assumed to be constant-time, following the approach of Schoenmakers~\cite{schoenmakers2024amortized}.

\begin{figure}[t]
  \centering
  \begin{minipage}[t]{0.48\textwidth}

\begin{lstlisting}
delete_min :: Tree Base -> Tree Base
delete_min x = match x with
  | leaf       -> leaf
  | node t a u -> meld t u

meld :: (Tree Base * Tree Base) -> Tree Base
meld x y = match x with
  | leaf       -> y
  | node t a u -> match y with
    | leaf       -> (node t a u)
    | node v b w -> if a <= b
      then bal t a
        (~ meld (node v b w) u)
      else bal v b
        (~ meld (node t a u) w)
\end{lstlisting}
  \end{minipage}\hfill
  \begin{minipage}[t]{0.48\textwidth}
\begin{lstlisting}
(* Skew Heap *)
bal :: (Tree Base  * Base * Tree Base) -> Tree Base
bal t a u = node u a t

(* Weight-Biased Leftist Heap *)
bal t a u =
  if weight t <= weight u
  then (node u a t)
  else (node t a u)

(* Rank-Biased Leftist Heap *)
bal t a u =
  if rank t <= rank u
  then (node u a t)
  else (node t a u)
\end{lstlisting}
  \end{minipage}
  \caption{Purely functional leftist heap implementation with three different balancing strategies implemented in our core language~(\lstinline{\~ e} stands for $\tick[1]{e}$).}
  \label{fig:leftist-code}
\end{figure}

\subsection{Automated Amortised Resource Analysis}

In the following, we introduce and motivate our formulation of AARA, which generalises the approach of Leutgeb et al.~\cite{hofmann2022type, leutgeb2021atlas, leutgeb2022automated}. The central idea of AARA is to infer, for each function $f$ a signature of form
\begin{equation*}
  \typed{f}{\atypdcl{(\typed{\vec{x}}{\vec{\alpha}})}{\Psi(\vec{x})}{\beta}{\Phi(\val)}} \tpkt
\end{equation*}
Uppercase Greek letters range over \emph{type annotations}. Such a signature expresses that $f$ given arguments $\vec{x}$ with potential $\Psi(\vec{x})$, produces the result $\val$ with potential $\Phi(\val)$. The function $\Psi$ is not the same as $\Phi$ since it additionally includes the amortised costs of $f$. We use the function \texttt{swap} as a running example to introduce our formalism. It captures the fundamental operation of skew heaps: swapping the left and right subtrees along the rightmost path. Its implementation is shown in the following listing, where the tick operator \lstinline{~} denotes one unit of cost per recursive function call.
\begin{lstlisting}
swap :: Tree Base -> Tree Base
swap x = match x with
| leaf       -> leaf
| node t a u -> let u' = ~ swap u in node u' a t
\end{lstlisting}
Our analysis infers the following type signature:
\begin{equation*}
  \typed{\texttt{swap}}{\atypdcl{\typed{x}{\Tree}}{\log\abs{x} + \phi(x)}{\Tree}{\phi(\val)}}
\end{equation*}
This corresponds to the standard function signature $\typed{\texttt{swap}}{\typed{x}{\Tree} \rightarrow \typed{\val}{\Tree}}$. Here, $\Tree$ denotes a tree data type over some element type. We write $\abs{x}$ for the number of leaves in tree $x$. The signature states that \texttt{swap} requires potential $\phi(\vec{x})$ (defined in Section~\ref{sec:aara}) and incurs amortised cost $\log\abs{x}$. While its worst-case cost is linear—since it descends the right spine—this amortised behaviour captures a similar structural effect that underlies the efficiency of skew heaps. Function \texttt{swap} and its tight analysis appeared previously as Proposition~3 accompanying~\cite{Sch92}. 

To realise this idea formally, we define the following form of type judgements, whose invariant reflects a generalisation of the classical potential method~\cite{tarjan1985amortized} to a pair of functions $\Psi$ and $\Phi$, following~\cite{hofmann2022type}:
\begin{equation*}
 \text{Given}\quad \eval{\sigma}{c}{e}{v},\quad\text{if}\quad \tjudge{\Gamma}{\Psi}{e}{\alpha}{\Phi} \quad\text{then}\quad \potAlt{\Gamma\sigma} \geqslant c + \pot{v} .
\end{equation*}
In words, given a substitution $\sigma$ from variables to values, if expression $e$ evaluates to value $v$, and we derive the type annotations $\Psi$ and $\Phi$, the amortised cost and potential $\Psi(\Gamma\sigma)$ of the typing environment $\Gamma$ are sufficient to cover both the evaluation cost $c$, and the resulting potential $\Phi(v)$.

To analyse a program we derive a signature $\typed{\vec{x}}{\vec{\alpha}}|\Psi \rightarrow \beta|\Phi$ for every function $f$, from a judgement $\tjudge{\typed{\vec{x}}{\vec{\alpha}}}{\Psi}{e}{\beta}{\Phi}$ for the function body $e$. The type derivation proceeds by applying type rules, according to our mostly syntax directed inference algorithm. For every rule application, the algorithm collects constraints involving type annotations, resulting in a linear real arithmetic constraint system, to be solved by Z3.

\paragraph{Generic Type System.} The analysis of leftist heaps requires the consideration of different families of potential functions.
Thus, in contrast to previous approaches, where the potential function is hard-coded into the type system, we present a generic type system from which we instantiate a type inference algorithm for specific choices of potential functions.
Based on this design, our implementation is compositional and modular: it recovers the expressive power of $\atlas$~\cite{leutgeb2021atlas} while enabling additional families of potential functions with minimal effort.
We describe our potential functions for different leftist heap variants in Section~\ref{sec:case-studies}.

\paragraph{Data Structure Invariants.} Reasoning about leftist heaps, it is essential to keep track of their invariants. To that end, we refine the existing base data types---particularly tree types---by replacing their conventional inductive definitions with refined inductive predicates that explicitly encode the invariant. These refinements permit the type system to track and enforce the properties required by leftist heap invariants throughout the analysis. For example a weight-biased tree will be represented by the refinement $W(\Tree) := \{ \leaf \} \mid \{\tree{t}{a}{u} \mid \#{t} \geq \# u\}$, where the weight $\#x$ of tree $x$ is defined as $\#\leaf = 0$ and $\#(\tree{t}{a}{u}) = \#t + \#u + 1$. A rank-biased tree is expressed analogously, where the rank of tree $x$ is defined as $\dag \leaf=0$ and $\dag(\tree{t}{a}{u}) = \dag u + 1$.

\paragraph{Path-Sensitivity.} The type system presented in~\cite{hofmann2022type, leutgeb2021atlas} includes the rule $\ruleite$, which applies to if-then-else statements. Their analysis is path-insensitive: it disregards the branching condition and therefore treats both branches symmetrically. We refine this rule by adding the branch condition to the typing context. This enhancement allows the analysis to exploit the additional information provided by the condition and to perform more precise, path-sensitive reasoning.

As an example consider the \texttt{bal} function for weight-biased leftist heaps shown in Figure~\ref{fig:leftist-code}. Suppose that the potentials of $\tree{u}{a}{t}$ and $\tree{t}{a}{u}$ differ --- that is, the potential function is asymmetric. In particular, assume that we are given potential $\log\abs{u}$, but in the first branch we require potential $\log\abs{t}$ due to the rotation. Crucially, the branch condition provides additional information: when \lstinline{weight t <= weight u} holds, we can derive $\log\abs{t} \leqslant \log\abs{u}$, ensuring that the available potential suffices to cover the potential of the result.

\paragraph{Template Parameters.}
So far, we have only seen type annotations involving program variables, such as~$\log\abs{x}$ for $\typed{x}{\Tree} \in \Gamma$. We now additionally allow free variables, denoted by uppercase letters. They are analogous to the \emph{ghost variables} known from Hoare logic~\cite{reynolds1998theories}. This extension enables the analysis to derive universally quantified statements—i.e., function signatures—that can later be instantiated with concrete program variables when needed. As an example, consider the signature $\atypdcl{\typed{x}{\Tree}}{\log(\abs{x} + \abs{X})}{\Tree}{\log(\abs{\val} + \abs{X})}$, corresponding to the property $\forall X.\ \log(\abs{x} + \abs{X}) \geq \log(\abs{\val} + \abs{X})$, which can be instantiated as $\log(\abs{x} + \abs{z}) \geq \log(\abs{\val} + \abs{z})$ for instance. We require template parameters to support function composition for the piecewise potential function of skew heaps, as we will demonstrate in Section~\ref{sec:case-studies}.


\section{Language and Cost Semantics}
\label{sec:semantics}
\begin{figure}[t!]
\centering
\begin{align*}
  v & \Coloneqq \false \mid \true && \mid \const\ \vec{v} \\ 
  \circ & \Coloneqq \textup{\lstinline{<}}
      \mid \textup{\lstinline{<=}}
      \mid \textup{\lstinline{>}}
      \mid \textup{\lstinline{>=}}
	\mid \textup{\lstinline{=}}
  \\
  e & \Coloneqq x 
                               && \mid v \\
    & \const\ \vec{x}
                               && \mid \error\\
    & \mid f~x_1~\dots~x_n
    && \mid \tick[$\color{black}a$]{e}\\
    & \mid e_1 \circ e_2
    && \mid \cif\ e_1\ \cthen\ e_2\ \celse\ e_3\\
    & \mid \match\ x\ \with
       \{\textup{\lstinline{|}} v \arrow e\}
    && \mid \vlet\ x~\equal~e_1\ \vin\ e_2\\
\end{align*}
\vspace{-0.9cm}
\caption{Abstract syntax of our (first-order) core language. Expressions are denoted by $e$, variables by $x$, and values by $v$.}
\label{fig:syntax}
\end{figure}
We define a small core language used to implement our benchmarks and its corresponding underlying cost semantics. The language shown in Figure~\ref{fig:syntax} is a pure first order functional language, based on the specification given by Leutgeb et al.~\cite{leutgeb2021atlas}.
In addition to the usual ML-like language primitives, like pattern matching and let bindings, we include a tick operator $\tick[a]{e}$ that does not change the value of the wrapped expression, but rather attributes a non-negative rational cost of $a \in \Qplus$ to the evaluation of this expression. We assume that each constructor $\const$ is associated with a fixed type of the form $\vec{\alpha} \to \alpha$. The expression $\error$ halts evaluation and inhabits every type, to allow partial implementations. 

We consider the following types. First, booleans ($\Bool$) with values $\true$ and $\false$. Second, an abstract base type $\m{Base}$, abbreviated as $\Base$, which carries no potential. Finally, trees ($\m{Tree}\ \m{Base}$), abbreviated $\Tree$, given by the constructors $\typed{\leaf}{\Tree}$ and $\typed{\tree{t}{a}{u}}{\Tree \times \Base \times \Tree \rightarrow \Tree}$. 

Our cost semantics is given in the style of big-step operational semantics with a call-by-value evaluation strategy. We define judgements $\eval{\sigma}{c}{e}{v}$, expressing that, given an environment (substitution) $\sigma$ mapping variable to values, the expression $e$ evaluates to value $v$ with costs $c$ (see Figure~\ref{fig:semantics}).
\begin{figure}[t]
\def\MathparLineskip{\lineskip=3.5pt}
\begin{mathpar}
\input{rules/sem/var.tex}
\and
\input{rules/sem/const.tex}
\and
\input{rules/sem/tick.tex}
\and
\input{rules/sem/app.tex}
\and
\input{rules/sem/let.tex}
\and
\input{rules/sem/itetrue.tex}
\and
\input{rules/sem/itefalse.tex}
\and
\input{rules/sem/match.tex}
\end{mathpar}
Here, $\sigma[x \mapsto w]$ represents the environment $\sigma$ updated so that $\sigma[x \mapsto w](x) = w$, while all other variable mappings remain unchanged. For function application, we define $\sigma'$ as $\sigma' := \{y_1 \mapsto x_1\sigma, \dots, y_n \mapsto x_n\sigma\}$ and for the rule covering \lstinline{match} $\sigma'' = \sigma \uplus \{x_1 \mapsto v_1,\dots, x_n \mapsto v_n\}$.
\caption{Evaluation rules for the big-step cost semantics.}%
\label{fig:semantics}
\end{figure}

\section{Declarative Typing for Amortised Analysis}
\label{sec:aara}
\subsection{Resource Type Judgements}
Our analysis is formulated using type judgements of the form $\tjudge{\Gamma;\guard}{\Psi}{e}{\alpha}{\Phi}$. The context $\Gamma$ is a standard typing environment that maps program variables to their types. The component $\guard$ is a collection of \emph{guard predicates} associated with the context; these predicates capture assumptions about the program state and will be defined precisely when we introduce data-structure invariants. The annotations $\Psi$ and $\Phi$ denote \emph{resource functions}, which map valuations of variables to real numbers.
 Formally, the left-hand side $\Gamma;\guard{\mid}\Psi$ should be read as $\Gamma;\guard{\mid}\Psi(\dom(\Gamma))$, where the resource function $\Psi$ is applied to the variables in the typing context to compute their potential.
For brevity, we write $\Psi(\Gamma)$ to denote $\Psi(\dom(\Gamma))$. The right-hand side of the judgement is annotated with the resource function $\Phi$. The judgement asserts that the expression $e$ has type $\alpha$ and that, given initial potential $\Psi(\Gamma)$, its evaluation yields a result value, represented by the variable $\val$ with potential $\Phi(\val)$.

\subsection{Well-Typed Programs} A program $P$ is a finite set of function definitions. Our analysis computes, for each function $f(\vec{x}) = e \in P$, two sets of signatures: $\mathcal{F}(f)$ and $\mathcal{F}^{\text{cf}}(f)$. We write $\typed{\vec{x}}{\vec{\alpha}}$ as shorthand for the sequence of bindings $\typed{x_1}{\alpha_1},\dots,\typed{x_n}{\alpha_n}$. For each \emph{costed signature} $\atypdcl{\typed{\vec{x}}{\vec{\alpha}}}{\Psi}{\beta}{\Phi} \in \mathcal{F}(f)$, we require the type judgement $\tjudge{\typed{\vec{x}}{\vec{\alpha}}}{\Psi}{e}{\beta}{\Phi}$ to hold.
These signatures account for tick statements and are required to compute amortised costs. In addition we compute \emph{cost-free signatures} $\atypdcl{\typed{\vec{x}}{\vec{\alpha}}}{\Psi}{\beta}{\Phi} \in \mathcal{F}^{\text{cf}}(f)$, such that the type judgement $\tjudgenacf{\typed{\vec{x}}{\vec{\alpha}}|\Psi}{e}{\beta}{\Phi}$ holds. The cost-free type judgement $\vdash^{\text{cf}}$ does not account for costs. Formally, it coincides with the ordinary typing judgement applied to the expression obtained by erasing all $\tick{}$ annotations. This judgement can therefore be used to establish auxiliary properties required by the analysis, such as emulating a size analysis; we illustrate this use in our case studies. A program satisfying these conditions is called \emph{well-typed}. 

\begin{figure}[t]
  \begin{mathpar}
    \input{rules/type/var.tex}
    \and
    \input{rules/type/const.tex}
    \and
    \input{rules/type/match.tex}
    \and
    \input{rules/type/ite.tex}
    \and
    \input{rules/type/let.tex}
    \and
    \input{rules/type/app.tex}
    \and
    \input{rules/type/tickd.tex}
\end{mathpar}
\caption{Syntax-directed type rules.}
\label{fig:type-syntax}
\end{figure}

\subsection{Resource Type System} Next, we present the type system that gives rise to the typing judgements introduced above. We distinguish between the syntax directed rules given in Figure~\ref{fig:type-syntax} and the structural rules given in Figure~\ref{fig:type-struct}. These rules provide a generic abstraction of the type system from~\cite{hofmann2022type, leutgeb2021atlas} that applies uniformly to arbitrary data types and potential functions; our extensions are \changed{\text{highlighted in grey}}. While the syntax-directed rules involve only equality constraints—ensuring that potential and cost are always preserved—we also require a set of structural rules. In contrast they can modify or discard potential and may, in principle, be applied at any point in the type derivation.

\paragraph{Example Derivation.} To develop an intuition for the type rules, we give a step-by-step explanation of the type derivation of $\texttt{swap}$, based on the potential function $\phi$ defined as $\phi(\leaf) = 1$ and $\phi(\tree{t}{a}{u}) = \phi(t) + \sfrac{1}{2}\log\abs{u} + \phi(u)$.

\paragraph{Deconstructing Potential.} Since the function body consists of a match expression, we first apply the corresponding rule $\rulematch$, that ensures the potential of the matched variable is distributed to the pattern variables.
\begin{equation*}
  \inferrule*[right=\rulematch]{
     \tjudgeml{\typed{t}{\Tree}, \typed{u}{\Tree} \mid \log(\abs{t} + \abs{u})
       + \phi(t) + \sfrac{1}{2}\log\abs{u} + \phi(u)}{\vlet\ u' = \text{\lstinline{\~ swap u}}\ \vin \dots}{\Tree}{\phi(\val)} \\
     \dots
  }
  {
    \tjudgeml{\typed{x}{\Tree} \mid \log\abs{x} + \phi(x) }{
      \match\ x\
      \with \text{\lstinline{|}} \tree{t}{a}{u} \arrow\ \dots\
      \text{\lstinline{|}} \leaf \arrow\ \dots
    }{\Tree}{\phi(\val)}
  }
\end{equation*}
For brevity we only consider the $\flstc{node}$ case, which has the obligation $\Psi(\typed{x}{\Tree}) = \Psi_{\flstc{node}}(\typed{t}{\Tree}, \typed{a}{\Base},\typed{u}{\Tree})$. In this instance we can easily verify that the equation $\log\abs{x} + \phi(x) = \log(\abs{t} + \abs{u}) + \phi(t) + \sfrac{1}{2}\log\abs{u} + \phi(u)$ holds.

\begin{figure}[t]
\begin{mathpar}
\input{rules/type/wvar.tex}
\and
\input{rules/type/shift.tex}
\and
\input{rules/type/share.tex}
\and
\input{rules/type/w.tex}
\end{mathpar}
\caption{Structural type rules.}\label{fig:type-struct}
\end{figure}

\paragraph{Expert Knowledge.}
\label{sec:aara-expert}
We apply the weakening rule $\rulew$ to continue the type derivation of \texttt{swap}, reducing the left-hand side potential of the judgement.
\begin{equation*}
  \inferrule*[right=\rulew]{
    \tjudgena{\typed{t}{\Tree}, \typed{u}{\Tree} \mid \sfrac{1}{2}\log\abs{t} + \log\abs{u} + 1 + \phi(t) + \phi(u)}{\vlet\ u' = \dots}{\Tree}{\phi(\val)}
  }
  {
    \tjudgena{\typed{t}{\Tree}, \typed{u}{\Tree} \mid \log(\abs{t} + \abs{u})
       + \phi(t) + \sfrac{1}{2}\log\abs{u} + \phi(u)}{\vlet\ u' =  \dots}{\Tree}{\phi(\val)}
   }
 \end{equation*}
 This step relies on the inequality $\log(\abs{t} + \abs{u}) \geq \sfrac{1}{2}\log\abs{t} + \sfrac{1}{2}\log\abs{u} + 1$, where the constant term $1$ is crucial to cover the tick cost of the recursive call later on. The weakening rule is the cornerstone of this analysis technique, as it allows incorporation of \emph{expert knowledge}~\cite{hofmann2022type, leutgeb2021atlas}, i.e., axioms about resource functions. 

 For example, the above inequality is a special case of the following lemma—an instance of Jensen’s inequality—also essential in the analysis of splay trees~\cite{leutgeb2021atlas}.
\begin{lemma}
\label{lem:favorite-log-lemma}
 Let $x,y \geq 1$ then $\log(x + y) \geq \sfrac{1}{2}\log x + \sfrac{1}{2}\log y + 1$.
\end{lemma}
\paragraph{Composition.} The next step in the type derivation of $\texttt{swap}$ is an application of the rule $\rulelet$.
\begin{equation*}
  \inferrule*[right=\rlabel{\rulelet}]{
    \tjudgeml{\typed{u}{\Tree} \mid \log\abs{u} + \phi(u) + 1}{\text{\lstinline{\~ swap u}}}{\Tree}{\phi(\val)}
    \\
    \tjudgeml{\typed{l}{\Tree}, \typed{x}{\Tree} \mid \phi(u') + \sfrac{1}{2}\log\abs{t} + \phi(t)}{\tree{u'}{a}{t}}{\Tree}{\phi(\val)}
  }{
    \tjudgeml{\typed{t}{\Tree}, \typed{u}{\Tree} \mid
      \sfrac{1}{2}\log\abs{t} + \log\abs{u} + 1 + \phi(t) + \phi(u)}{\vlet\ u' = \text{\lstinline{\~ swap u}}\ \vin \dots}{\Tree}{\phi(\val)}
  }
\end{equation*}
 
Our language requires the result of each function call to be bound via a let-binding, which is essential for correctly distributing potential during function composition. The rule splits the potential between the binding expression and the body according to the variables in the contexts $\Gamma$ and $\Delta$, by decomposing the potential as $\Psi(\Gamma,\Delta) = \Psi_1(\Gamma) + \Psi_{\Delta}(\Delta) + \Psi_{\text{mix}}(\Gamma,\Delta)$ and $\Psi_2(\typed{x}{\alpha},\Delta) = \Phi_1(e_1) + \Psi_{\Delta}(\Delta) + \Omega(\typed{x}{\alpha},\Delta)$, where $\Psi_{\text{mix}}(\Gamma,\Delta)$ contains terms involving both $\Gamma$ and $\Delta$. Since this mixed potential can not be passed through the binding directly, we relate it to the body potential by requiring $\Psi_{\text{mix}}(\Gamma,\Delta) \geq \Omega(\typed{x}{\alpha},\Delta)$. This constraint is discharged by a set of cost-free typings.

For \texttt{swap} the binding expression only contains the variable $u$, so all potential terms involving $u$ are assigned to the type judgement for the binding. The body is then typed with the potential of the remaining variables and the newly introduced potential of variable $u'$.


The type judgement for the binding is resolved by an application of $\ruletick$, which adds the cost to the right-hand side of the judgement, followed by $\ruleshift$ which removes the costs from both sides and finally an application of the rule $\ruleapp$, governing the function call. In its simplest form $\ruleapp$ ensures that both the costs and potential match the signature of the called function.
\begin{equation*}
  \inferrule*[right=\rlabel{\ruletick}]{
    \inferrule*[right=\rlabel{\ruleshift}]{
      \inferrule*[right=\rlabel{\ruleapp}]{
        \atypdcl{\typed{x}{\Tree}}{\log\abs{x} + \phi(x)}{\Tree}{\phi(\val)} \in \mathcal{F}(\texttt{swap})
      }{
        \tjudge{\typed{u}{\Tree}}{\log\abs{u} + \phi(u)}{\text{\lstinline{swap u}}}{\Tree}{\phi(\val)}
      }
    }{
      \tjudge{\typed{u}{\Tree}}{\log\abs{u} + \phi(u) + 1}{\text{\lstinline{swap u}}}{\Tree}{\phi(\val) + 1}
    }
  }{
    \tjudge{\typed{u}{\Tree}}{\log\abs{u} + \phi(u) + 1}{\text{\lstinline{\~ swap u}}}{\Tree}{\phi(\val)}
  }
\end{equation*}
The body of the \texttt{let} expression is a single \lstinline{node} expression that constructs the result tree. The derivation therefore concludes with an application of rule $\ruleconst$, which ensures that the constructor arguments carry sufficient potential, that is, $\phi(u') + \sfrac{1}{2}\log\abs{t} + \phi(t) = \phi(\tree{u'}{a}{t})$. 
\begin{equation*}
  \inferrule*[right=\ruleconst]{ }{
        \tjudge{\typed{t}{\Tree}, \typed{u'}{\Tree}}{\phi(u')
            + \sfrac{1}{2}\log\abs{t}
            + \phi(t)}{\tree{u'}{a}{t}}{\Tree}{\phi(\val)}
        }
      \end{equation*}

\paragraph{Data Structure Invariants.} 
Beyond their primary role, the rules $\ruleconst$ and $\rulematch$ enforce data structure invariants, which are supported as refinements of base data types by predicates. We define a refinement of type $\alpha$ as
\begin{equation*}
  R(\alpha) = \{ \const_1 \vec{x}_1 \mid p_{\alpha,\const_1}(\vec{x}_1)\} \mid \dots \mid \{\const_n \vec{x}_n \mid p_{\alpha,\const_n}(\vec{x}_n)\}
\end{equation*}
where each constructor $\const_i$, with parameters $\vec{x}_i$ is associated with a predicate $p_{\alpha,\const_i}(\vec{x}_i)$ describing its invariant. We write $R$ for $R(\alpha)$ when the base type is clear from the context. Predicates are defined by the following grammar:
\begin{mathpar}
  p(\vec{x}) \Coloneq \top \mid m\,x_i\; \circ\; m\,x_j \quad (x_i,x_j \in \vec{x})\and
  m \Coloneq \# \;\mid\; \dag 
\end{mathpar}
While this definition suffices for the predicates required in our analysis, the approach could be straightforwardly generalised to cover a wider range of predicates. We write $p(\vec{x})$ for the instantiation of a predicate with the variables of a matched pattern. The rule $\rulematch$ extends the typing context with the predicate $p_{\alpha,\const_i}(\vec{x_i})$ for each matched pattern $\const_i \vec{x_i}$. Unrefined constructors are treated as if they were refined with the trivial predicate $\top$, which is omitted from the context, i.e. $\guard, \top = \guard$. Conversely, $\ruleconst$ requires the predicate associated with the applied constructor to be present in the type context, thereby restricting construction to refined values.

\paragraph{Path-Sensitive Typing of Conditionals.} We write $\wfguard{p}$ to denote $p$ if it is a well-formed predicate generated by the grammar above, and $\top$ otherwise. If the branching condition $e_1$ of an if-then-else statement is a well-formed predicate, we allow the rule $\ruleite$ to add it to the typing context, so it can be used by the rule $\rulew$ for path-sensitive reasoning. The potential of expression $e_1$ is zero, since we do not assign potential to Boolean expressions. We do, however, allow potential to be spent when evaluating $e_1$, as it may, for instance, be a function call with a potential-bearing argument.

\subsection{Soundness} Our main soundness theorem\ifextended, which we prove in Appendix~\ref{sec:soundness-proof},\fi\ guarantees that any valid typing respects the basic inequality underlying the potential method, given that the set of guard predicates $\mathcal{G}$ holds. Since the theorem relates typing judgements to concrete executions via substitutions for program variables, we extend substitutions pointwise to environments, writing $\Gamma\sigma$ for the environment obtained by applying $\sigma$ to all types in $\Gamma$.
\begin{restatable}{theorem}{soundnessTheorem}{Soundness Theorem.} Let $P$ be well-typed and $\sigma$ be a substitution for program variables. Suppose $\tjudge{\Gamma; \mathcal{G}}{\Psi}{e}{\alpha}{\Phi}$ and $\eval{\sigma}{c}{e}{v}$. Then, we have $\bigwedge \guard \Rightarrow \Psi(\Gamma\sigma) \geq c + \Phi(v)$. Further, if $\tjudgecf{\Gamma}{\Psi}{e}{\alpha}{\Phi}$, then $\bigwedge \guard \Rightarrow \Psi(\Gamma\sigma) \geq \Phi(v)$.%
  \label{t:soundness}
\end{restatable}

\section{Algorithmic Type System and Resource Inference}
\label{sec:inference}

The declarative type system we have seen in the previous section established the validity of concrete resource functions. Next we turn our attention to the problem of reconstructing these resource functions from scratch, i.e.~inferring the required type annotations. Instead of employing type variables that can stand for any function, we restrict ourselves to \emph{resource templates}, that characterise a family of resource functions by undetermined coefficients. 

\subsection{Resource Templates} 
Let $\tlang{}(\vec{x})$ (the \emph{template language}) denote a generating function for a set of non-negative \emph{template terms}, i.e.~syntactic arithmetic expressions over the set of variables $\vec{x}$. Examples of template terms include $\log |x|$, $\log(|x|+1)$, $\log \abs{\val}$ or $\log(|x|+|y| + \abs{X})$. For each term $t \in \tlang{}(\vec{x})$, we associate a fresh symbolic coefficient $q_t$. These coefficients are then combined linearly to form symbolic \emph{resource templates}.

\begin{definition}[Resource template]
  Let $\vec{x}$ be a set of variables, called the arguments of the resource template. The resource template \(Q(\vec{x})\) is defined as
  \begin{equation*}
    Q_{\tlang{}}(\vec{x}) \coloneq \sum_{t \in \tlang{}(\vec{x})} q_t \cdot t 
  \end{equation*}
  where each \(q_t\) is a symbolic coefficient; we omit the subscript \(\tlang{}\) when the underlying template language is clear from the context.
\end{definition}

\begin{example}\label{ex:sum-of-logs} We instantiate the abstract notion of a template language by defining a concrete \emph{sum-of-logs} template language, which generalises the templates used in the analysis of splay trees~\cite{leutgeb2021atlas,leutgeb2022automated}. 
\begin{equation*}
  \tlangSol(\vec{x}) \coloneq \tlangLog(\vec{x}) \cup \{\phiSol(x) \mid x \in \vec{x} \}
\end{equation*}
Let $\params = (a,b,c)$ then $\phiSol(\leaf) = 1$, $\phiSol(\tree{t}{a}{u}) = \phiSol(t) + \psiSol(t,u)  + \phiSol(u)$
where $\psiSol(t,u) = a\log\abs{t} + b\log\abs{u} + c\log(\abs{t} + \abs{u})$ and
\begin{equation*}
  \tlangLog(\vec{x}) \coloneq \{\log(\vec{a}\abs{\vec{x}} + c) \mid \vec{a} \in \{0,1\}^{\abss{\vec{x}}}, c \in \{-1,0,1,2\}, \textstyle\sum\vec{a} + c \geq 1\}
\end{equation*}
We define $\vec{a}\abs{\vec{x}} = \sum_i a_i\abs{x_i}$, where $\abs{x}$ denotes number of leaves for tree $x$ and we write $\abss{\vec{x}}$ for the cardinality of $\vec{x}$. 
In the following we abbreviate $q_{\phiSol(x)}$ as $q_x$ and $q_{\log(\vec{a}\abs{\vec{x}} + c)}$ as $q_{(\vec{a}\vec{x},c)}$. 
\end{example}
  Note that the terms of $\tlangSol$ are only valid when applied to arguments of type $\Tree$. Each template language is associated with an argument type; unless stated otherwise, we assume this type to be $\Tree$. The analysis of \texttt{swap} is based on $\tlangSol$ where we choose $\params = (0,\sfrac{1}{2},0)$, i.e. $\phiSol(\tree{t}{a}{u}) = \phiSol(t) + \sfrac{1}{2}\log\abs{u} + \phiSol(u)$. To express the cost and potential $\log\abs{x} + \phiSol(x)$ with template function $\Psi(\typed{x}{\Tree})$, we set $q_{(x)} = q_{x} = 1$ and all other coefficients to zero.

\subsection{Algorithmic Type Rules}


We let $C$ range over sets of linear real arithmetic constraints. The constraint typing relation is written $\tjudgect{\Gamma;\guard}{Q}{e}{\alpha}{Q'}{C}$, where in contrast to the declarative relation annotations $Q$ and $Q'$ are resource templates rather than concrete resource functions. Again $\Gamma;\guard \mid Q$ should be read as $\Gamma;\guard \mid Q(\dom(\Gamma))$. As a convention we assume that coefficients not mentioned in $C$ are constrained to be zero.

A model $M$ for a constraint system is an assignment that maps its variables to real numbers. Given a model $M$, we instantiate a template by substituting each symbolic coefficient according to $M$, obtaining a concrete resource function.
\begin{definition}[Model instantiation of a resource function]
Let $\vec{x}$ be a set of variables, \(Q(\vec{x}) = \sum_{t \in \tlang{}(\vec{x})} q_t \cdot t\) be a resource template, $M$ be a model and $M(q_t) \in \Q$, we then define \emph{model instantiation} as a substitution
  \begin{equation*}
    Q^M(\vec{x}) \coloneq \sum_{t \in \tlang{}(\vec{x})} M(q_t) \cdot t \tpkt
  \end{equation*}
\end{definition}

Instead of writing out the algorithmic type rules, we give a construction, parameterised in the template language $\tlang{}$, that allows to transform any declarative rule, into the corresponding algorithmic rule.

We let $R$ range over declarative rules and construct for every rule $R$ an algorithmic rule $\ctrule{R}$:
\begin{equation*}
  \inferrule*[right=\rlabel{\ctrule{R}}]{
    \tjudgect{\Gamma_1}{Q_1}{e_1}{\alpha}{Q'_1}{C_1}
    \dots
    \\
    \tjudgect{\Gamma_n}{Q_n}{e_n}{\alpha}{Q'_n}{C_n}
  }{
    \tjudgect{\Gamma}{Q}{e}{\alpha}{Q'}{\bigcup_i C_i \cup \ctr(Q,Q',\vec{Q},\vec{Q'}, \vec{\Gamma}) }
  }
\end{equation*}
The algorithm chooses the annotations for the sub-derivations in $\vec{Q}$ and $\vec{Q'}$ fresh, such that they are distinct from the annotations $Q$ and $Q'$ and their arguments match the given context. The construction is fully characterised by the set $\cgen$ of constraint generators $\ctr(Q,Q',\vec{Q},\vec{Q'},\vec{\Gamma})$. For each rule $R$, these functions produce constraints over the coefficients of the supplied templates that guarantee the obligations on resource functions from the declarative rule; this set $\cgen$ thus defines a concrete instance of the algorithmic type system for the template language $\tlang{}$. For the definition of $\ctr(Q,Q',\vec{Q},\vec{Q'},\vec{\Gamma})$ we use an inference-rule presentation. The function application in the conclusion is defined by the set of all constraints listed in the premises, implicitly quantified over any side conditions. Parameters written as $\_$ are ignored by the corresponding constraint generator and impose no constraints on the associated annotations.
\begin{example}
  For the sum-of-logs template language $\tlangSol$, $\rulematch$ has the following algorithmic form.
  \begin{equation*}
    \inferrule*[right=\rlabel{\ctrule{Match}}]{
      \tjudgectgen{\Gamma}{P}{e_1}{\alpha}{Q'}{C_{2}}{\tlangSol}
      \\
      \tjudgectgen{\Gamma,\typed{t}{\Tree},\typed{u}{\Tree}}{R}{e_2}{\alpha}{Q'}{C_1}{\tlangSol}
    }{
      {\begin{array}{ll}
        \Gamma,\typed{x}{\Tree}{\mid}Q\,\vdash_{\tlangSol}\,\flstk{match }x\flstk{ with } &\flst{| }\leaf\flst{ -> }e_1\\
        &\flst{| }\tree{t}{a}{u}\flst{ -> }e_2\colon\!{\alpha}{\mid}Q';C
      \end{array}}
  }
\end{equation*}
where $C = C_{1} \cup C_{2} \cup \ctr*{\tlangSol}{\labMatch}(Q,Q',P,R,(\Gamma,\typed{x}{\Tree}), (\Gamma,\typed{t}{\Tree},\typed{u}{\Tree}))$.
Note that we omit $\typed{a}{\Base}$ from the typing context in the node case because $\tlangSol$ is only defined on trees --- the elements themselves do not carry any potential. The constraints defining $\ctr*{\tlangSol}{\labMatch}$ are given in Figure~\ref{fig:sol-match-constraints} and guarantee that $P^M(\Gamma) = Q^M(\Gamma, \leaf)$ and $R^M(\Gamma, \typed{t}{\Tree}, \typed{u}{\Tree})) = Q^M(\Gamma, \tree{t}{a}{u})$ in accordance with the corresponding declarative rule.
\begin{figure}[t]
  \centering
  \begin{equation*}
  \inferrule*{
    p_{(\vec{a}\vec{x}, d)} = \displaystyle{\sum}_{b+c=d} q_{(\vec{a}\vec{x}, bx,c)} + q_x[\vec{a} = \vec{0}, d = 2] \\
    r_{(\vec{a}\vec{x}, bt, bu, c)} = q_{(\vec{a}\vec{x}, bx, c)} + f \cdot q_x[\vec{a} = \vec{0}, b = 1, c = 0] \\
    p_{x_i} = r_{x_i} = q_{x_i}\\
    r_t = d \cdot q_x \\ r_u = e \cdot q_x \\
    r_{(t)} = r_{(u)} = q_{(x)} \\
  }{
    \ctr*{\mathcal{L}_{\mathrm{sol},(d,e,f)}}{\labMatch}(Q,\_,P,R,(\Gamma,\typed{x}{\Tree}), (\Gamma,\typed{t}{\Tree},\typed{u}{\Tree}))
  }  
  \end{equation*}
  \caption{Constraints defining $\ctr*{\tlangSol}{\labMatch}(Q,Q',P,R,(\Gamma,\typed{x}{\Tree}), (\Gamma,\typed{t}{\Tree},\typed{u}{\Tree}))$, where $\vec{x} = \dom(\Gamma)$, $x_i \in \dom(\Gamma)$ and $[p]$ denotes the Iverson bracket for predicate $p$, i.e. $[p] = 1$ if $p$ holds and $[p] = 0$.}
  ~\label{fig:sol-match-constraints}
\end{figure}
\end{example}

Most of the functions in $\cgen$ can be defined generically, i.e.~independent of the given template language. In the rule $\rulevar$ for instance, we need to establish $\Psi(\typed{x}{\alpha}) = \Phi(\typed{x}{\alpha})$, which can be accomplished by coefficient-wise equality constraints. \ifextended The definition of these generators is given in Figure~\ref{fig:templ-agnostic-constraints} in Appendix~\ref{sec:constraints}.\fi

In fact only the constraints for $\ctrule{Const}$, $\ctrule{Match}$ and $\ctrule{Let}$ crucially depend on the choice of template language, as they encode how potential is modified by construction, deconstruction, and function composition. \ifextended We provide the complete set of constraints for $\tlangSol$ in Figure~\ref{fig:sol-constraints} in Appendix~\ref{sec:constraints}, where we also establish their soundness. They generalise the constraints presented in~\cite{hofmann2022type, leutgeb2021atlas} and allow to recover their type rules. \fi Note that the constraint-generating function $\ctr(c,Q,Q')$ for $\ruleconst$ additionally takes the constructor $c$ as a parameter, effectively yielding a distinct algorithmic rule for each constructor.

\paragraph{Weakening.} To generate the constraints for $\ctrule{W}$, we follow the approach of~\cite{hofmann2022type, leutgeb2021atlas}: inequalities of the form $\Psi(\vec{y}) \geq \Psi'(\vec{y})$ cannot be checked symbolically and are transformed into a set of existentially quantified linear constraints. To this end, the technique of~\cite{hofmann2022type} employs a variant of Farkas’ Lemma.
\begin{lemma}[Farkas' Lemma]%
\label{lem:farkas}
Suppose  $A\vec{x} \leqslant \vec{b}, \vec{x} \geqslant 0 $ is solvable. Then the following assertions are equivalent.\begin{inparaenum}[(i)]
\item $\forall \vec{x} \geqslant 0. \ A\vec{x} \leqslant \vec{b} \Rightarrow \vec{u}^T\vec{x} \leqslant \lambda$ and
\item $\exists \vec{f} \geqslant 0. \  \vec{u}^T \leqslant \vec{f}^T A \land \vec{f}^T \vec{b} \leqslant \lambda$.
\end{inparaenum}
\end{lemma}
In our setting, the vector $\vec{x}$ represents template terms. The constraints $A\vec{x} \leq \vec{b}$ encode template language specific expert knowledge. Besides Lemma~\ref{lem:favorite-log-lemma}, it also includes a partial order over logarithmic terms induced by the monotonicity of the logarithm\ifextended. We describe this order formally in Appendix~\ref{sec:guard-mono}, where we also explain how guard predicates are used to compute it. \fi The inequality $u^T\vec{x} \leq \lambda$ corresponds to the inequality of interest, i.e. $\Psi'(\vec{y}) -  \Psi(\vec{y}) \leq 0$.

\subsection{Soundness} We establish the soundness of the algorithmic type system in Theorem~\ref{t:soundness-algo}. To prove it, we rely on Lemma~\ref{lem:algo} that links the two type systems. Specifically, any solution to the constraints obtained by the algorithmic typing relation corresponds to concrete resource functions typable in the declarative system, allowing us to apply the soundness of the declarative system (Theorem~\ref{t:soundness}).

\begin{lemma}
  \label{lem:algo}
  Let $M$ be a satisfying model for a set of constraints $C$. If $\tjudgect{\Gamma;\guard}{Q}{e}{\alpha}{Q'}{C}$ then $\tjudge{\Gamma;\guard}{Q^M}{e}{\alpha}{Q'^M}$.
\end{lemma}
\begin{proof}
  By induction over $\tjudgect{\Gamma;\guard}{Q}{e}{\alpha}{Q'}{C}$ where we distinguish the last rule $\ctrule{R}$ in the derivation. \ifextended Using the lemmas provided in Appendix~\ref{sec:constraints} we show that the template constraints, given by $\ctr(Q,Q',\vec{Q},\vec{Q'},\vec{\Gamma})$, imply the obligations of the corresponding declarative rule.\fi
\end{proof}

\begin{theorem}{Soundness of the algorithmic type system.}
  \label{t:soundness-algo}
  Let $P$ be well-typed, $\sigma$ be a substitution for program variables and $M$ be a satisfying model for a set of constraints $C$. If $\tjudgect{\Gamma;\guard}{Q}{e}{\alpha}{Q'}{C}$ and $\eval{\sigma}{c}{e}{v}$ then $\bigwedge \mathcal{G} \Rightarrow Q^M(\Gamma)\sigma \geq c + Q'^M(v)$.
\end{theorem}

\subsection{A Full Cost Analysis} The algorithmic type rules almost resemble our real type inference algorithm, except for two details\ifextended , both of which are addressed in the description of our implementation in Appendix~\ref{sec:implementation}. \fi

First, the structural rules do not yield a deterministic algorithm. \ifextended We therefore apply them heuristically to keep inference tractable. \fi
Secondly the above formulation only demands any satisfying solution, but we are interested in annotations that correspond to the minimal amortised costs for a given function. Following~\cite{leutgeb2021atlas}, we encode this notion into a cost-function which we supply to the underlying optimising SMT-solver to obtain a minimal solution. 

\section{Instantiating the Analysis for Leftist Heaps}
\label{sec:case-studies}

\begin{table}[t]
  \centering
  \caption{Analysis instances for different leftist heap variations.}
  \label{tab:instances}
  \begin{tabularx}{\textwidth}{l>{\hspace{5ex}}Xr}
    \toprule
    Data Structure & Type & $\tlang{}$\\
    \midrule
    Skew Heap & $\Tree$ & $\tlangPw, \tlangSol$\\
    Weight Biased & $\WeightTree := \{ \leaf \} \mid \{\tree{t}{a}{u} \mid \#{t} \geq \# u\}$ & $\tlangPw, \tlangSol$\\
    Rank Biased & $\RankTree := \{ \leaf \} \mid \{\tree{t}{a}{u} \mid \dag{t} \geq \dag{u}\}$ & $\tlangRank$\\
    \bottomrule
  \end{tabularx}
\end{table}
We now instantiate the inference algorithm for the various leftist heap variants. Table~\ref{tab:instances} summarises, for each data structure, the potential-carrying data type, the associated template language, its constraint generators, and expert knowledge. \ifextended Appendix~\ref{sec:additional-case-studies} provides further intuition via selected type-derivation examples.\fi

We first present the template languages $\tlangPw$ and $\tlangSol$, corresponding to potential functions for skew heaps found in the literature: a piecewise potential due to Sleator and Tarjan~\cite{sleator1986self}, and a sum-of-logs potential introduced by Kaldewaij and Schoenmakers~\cite{kaldewaij1991derivation}. We then introduce the template language $\tlangRank$ for rank-biased heaps.\ifextended The soundness of the given constraint generating functions is established in the Appendix~\ref{sec:constraints}. \fi 

\subsection{Piecewise Potential}
\label{sec:piece-wise-templ}
We define the piecewise potential function $\phiPw(\leaf) = 0$ and $\phiPw(\tree{t}{a}{u}) = \phiPw(t) + [\abs{t} < \abs{u}] + \phiPw(u)$ using Iverson brackets, where $[p] = 1$ if the predicate $p$ holds and $[p] = 0$ otherwise. Iverson brackets are expressed directly in a new template language $\tlangPw$. This language also provides terms to model the function $\phiPw$ and logarithmic terms for costs.
\begin{align*}
  \tlangPw(\vec{x}) &\coloneq\; \tlangLog(\vec{x}) \cup \{\phiPw(x) \mid x \in \vec{x}\} \\
            &\cup \{[{\textstyle \sum_{y \in \vec{y}}} \abs{y} + c < {\textstyle \sum_{z \in \vec{z}}} \abs{z} ] \mid \vec{z},\vec{y} \subseteq \vec{x} \cup \{X\}, c \in \{-1,0,1,2\} \}
\end{align*}
$X$ is a fixed template parameter, i.e.~a variable that does not occur in the program. We use $[\vec{x} + c < \vec{y}]$ as shorthand for the Iverson bracket, involving the sum of weights of these variables and abbreviate $q_{\phiPw(x)}$ as $q_x$. The constraint generating functions for $\tlangPw$ are shown in Figure~\ref{fig:piece-constraints}.

\paragraph{Expert Knowledge.} We use the following lemma to relate Iverson brackets and logarithmic terms.
  \begin{lemma}
    \label{lem:logLemma2Iver}
    Let $x, y \geq 1$ then $\log(x + y) \geq \log y + [x \ge y]$ and \\
    \begin{enumerate*}[(i)]
    \item $\log(x + y) + [x < y] \geq \log y + 1$;
    \item $\log(x + y) \geq \log y + [x > y]$
    \end{enumerate*}.
  \end{lemma}
  For the analysis of weight-biased leftist heaps we require a variation of Lemma~\ref{lem:logLemma2Iver}, where the condition $x \geq y$ takes guard predicates into account.
  \begin{corollary}
    \label{lem:guarded-favorite-log}
  Let $x, y \geq 1, x \geq y$ then $\log(x + y) \geq \log y + 1$.
\end{corollary}
\begin{figure}[t]
  \begin{mathpar}
    \inferrule*{
      \ctr*{\tlangLog}{\labConst}(\flstc{leaf}, Q, Q')
    }{
      \ctr*{\tlangPw}{\labConst}(\flstc{leaf}, Q, Q')
    }
    \and
    \inferrule*{
      q_t = q_u = q_{[t < u]} = q'_{\val}\\
      q_{[\bar{x},t,u + c < \bar{y}]} = q'_{[\bar{x},\val + c < \bar{y}]}\\\\
      q_{[\bar{x} + c < \bar{y},t,u]} = q'_{[\bar{x} + c < \bar{y},\val]}\\
      \ctr*{\tlangLog}{\labConst}(\flstc{node}\ t\ a\ u,Q,Q')
    }{
      \ctr*{\tlangPw}{\labConst}(\flstc{node}\ t\ a\ u,Q,Q')
    }
    \and
    \inferrule*[vskip=0.5ex]{
      p_{[\bar{x} + c < \bar{y}]} = q_{[\bar{x},x + (c - 1) < \bar{y}]} +
      q_{[\bar{x} + (c + 1) < \bar{y},x]} + q_{[\bar{x} + c < \bar{y}]}\\
      r_{[\bar{x},t,u + c <\bar{y}]} = q_{[\bar{x},x + c < \bar{y}]}\\
      r_{[\bar{x} + c <\bar{y},t,u]} = q_{[\bar{x} + c < \bar{y},x]}\\
      r_{[\bar{x} + c <\bar{y}]} = q_{[\bar{x} + c < \bar{y}]}\\
      p_{x_i} = r_{x_i} = q_{x_i}\\
      \ctr*{\tlangLog}{\labMatch}
      (Q,Q',P,R,(\Gamma,\typed{x}{\Tree}),
      (\Gamma,\typed{t}{\Tree},\typed{u}{\Tree}))
    }{
      \ctr*{\tlangPw}{\labMatch}
      (Q,Q',P,R,(\Gamma,\typed{x}{\Tree})),
      (\Gamma,\typed{t}{\Tree},\typed{u}{\Tree})
    }
    \and
    \inferrule*[vskip=1ex]{
      q_{[\vec{x_1} + c < \vec{x_2}]} = p_{[\vec{x_1} + c < \vec{x_2}]}\\
      r_{i[x \mapsto \val]} = p'_i\\
      \forall c \neq 0.\ q_{[\vec{x} + \vec{y} + c < \vec{z}]} = p^{[\vec{x} + \vec{y} + c < \vec{z}]}_{[\vec{x} < \abs{X}]} +  p^{[\vec{x} + \vec{y} < \vec{z}]}_{[\vec{x} + c < \abs{X}]} \\
      r_{[x + \vec{y} + c < \vec{z}]} = p'^{[\vec{x} + \vec{y} + c < \vec{z}]}_{[\val + c < \abs{X}]} \\
      \ctr*{\tlangLog}{\labLet}
      (Q,Q',P,P',R,\vec{P},\vec{P'},\Gamma,
      (\Delta,\typed{x}{\Tree}))
    }{
      \ctr*{\tlangPw}{\labLet}(Q,Q',P,P',R,\vec{P}, \vec{P'}, \Gamma, (\Delta, \typed{x}{\Tree})) 
    }
  \end{mathpar}
  \caption{Constraint generating functions for $\tlangPw$ where $\vec{x} = \dom(\Gamma)$, $\vec{y} = \vec{z} = \dom(\Delta)$ and $\vec{x_1}, \vec{x_2} \subseteq \vec{x}$. $\bar{x},\bar{y}$ denote (possibly empty) sequences of variables.}
    \label{fig:piece-constraints}
  \end{figure}

\subsection{Sum-of-Logs Potential}
\label{sec:skew-sol} The template language \(\tlangSol\) is defined in Example~\ref{ex:sum-of-logs}\ifextended and its constraint generators are shown in Figure~\ref{fig:sol-constraints}\fi. We instantiate it for the analysis of skew heaps and weight-biased leftist heaps. The value of $\params$ is chosen manually as part of the template design. We experimented with different combinations reported in the literature, which affect the precision of the resulting bounds (see Figure~\ref{fig:results}).

\paragraph{Expert Knowledge.} To support the optimal golden-ratio bound for skew heaps and weight-biased leftist heaps, we generalise Lemma~\ref{lem:favorite-log-lemma} and obtain the following result.
\begin{lemma}
\label{lem:new-favorite-log-lemma}
  Let $a,b > 0$.
  Then, $(a+b)\log(x + y) \geq a \log x + b \log y + 1$ for all $x,y > 0$ if and only if $\frac{(a+b)^{a+b}}{a^ab^b} \ge 2$.
\end{lemma}
\begin{proof}
  Consider that $(a+b)\log(x + y) \geq a \log x + b \log y + 1$ is equivalent to $\frac{(x+y)^{a+b}}{x^ay^b} \ge 2$.
  Then, set $r = \frac{y}{x} > 0$ and observe that the function $(1+r)^a(1+\frac{1}{r})^b$ has a minimum at $r=\frac{b}{a}$.
\end{proof}
When minimising $2a+b$ in the bound for $\texttt{meld}$ under the boundary condition $\frac{(a+b)^{a+b}}{a^ab^b} \ge 2$, one obtains $a = \frac{\phi\ln{2}}{(\phi+2)\ln{\phi}} \approx 0.6441$ and $b = \frac{\ln{2}}{(\phi+2)\ln{\phi}} \approx 0.3981$, which matches the result of~\cite{kaldewaij1991derivation}. We note that verifying these optimal values for $a$ and $b$ would require a symbolic manipulation of algebraic numbers, which for the moment is outside of our approach.
Instead we contend ourselves with approximating these values by rational numbers (our choice is arbitrary and any approximation satisfying Lemma~\ref{lem:new-favorite-log-lemma}’s side condition suffices):
\begin{restatable}{lemma}{logIneqApprox}
$\left (\frac{105}{163} + \frac{3115}{7824}\right ) \log(x + y) \geq$
$\frac{105}{163} \log x  + \frac{3115}{7824} \log y + 1$ for $x,y \geq 1$.
\end{restatable}

\subsection{Rank Potential} For the analysis of rank-biased leftist heaps we introduce the following  template language, whose constraint generators are given in Figure~\ref{fig:rank-constraints}.
\begin{equation*}
  \tlangRank(\vec{x}) = \mathcal{L}_{\text{log}}(\vec{x}) \cup  \{\dag x \mid x \in \vec{x}\} 
\end{equation*}
\paragraph{Expert Knowledge.} We require only the well-known inequality $\dag{x} \le \log{\abs{x}}$~\cite{schoenmakers2024amortized}.

\begin{figure}[t]
  \centering
  \begin{mathpar}
    \inferrule*{
      \ctr*{\tlangLog}{\labConst}(\flstc{leaf}, Q,Q')
    }{
      \ctr*{\tlangRank}{\labConst}(\flstc{leaf}, Q, Q')
    }
    \and
    \inferrule*[vskip=0.5ex]{
      p_{\dag x_i} = r_{\dag x_i} = q_{\dag x_i}\\
      r_{(2)} = q_{\dag x} + q_{(2)}\\
      r_{\dag u} = q_{\dag x}\\\\
      \ctr*{\tlangLog}{\labMatch}
      (Q,Q',P,R,(\Gamma,\typed{x}{\Tree}),
      (\Gamma,\typed{t}{\Tree},\typed{u}{\Tree}))
    }{
      \ctr*{\tlangRank}{\labMatch}
      (Q,Q',P,R,(\Gamma,\typed{x}{\Tree}),
      (\Gamma,\typed{t}{\Tree},\typed{u}{\Tree}))
    }
    \and
    \inferrule*[vskip=0.5ex]{
      q_{\dag t} = 0,\,q_{\dag u} = q'_{\dag \val}\\
      q_{(2)} = q'_{\dag \val} + q'_{(2)}\\\\
      \ctr*{\tlangLog}{\labConst}(\flstc{node}\ t\ a\ u,Q,Q')
    }{
      \ctr*{\tlangRank}{\labConst}(\flstc{node}\ t\ a\ u,Q,Q')
    }
    \and
    \inferrule*{
      p_{\dag x_i} = q_{\dag x_i}\\
      r_{\dag y_i} = q_{\dag y_i}\\
      r_{\dag x} = p'_{\dag \val}\\\\
      \ctr*{\tlangLog}{\labLet}
      (Q,Q',P,P',R,\vec{P},\vec{P'},\Gamma,
      (\Delta,\typed{x}{\Tree}))
    }{
      \ctr*{\tlangRank}{\labLet}
      (Q,Q',P,P',R,\vec{P},\vec{P'},\Gamma,
      (\Delta,\typed{x}{\Tree}))
    }
  \end{mathpar}
   \caption{Constraint generating functions for $\tlangRank$, where $x_i \in \dom{\Gamma}$ and $y_i \in \dom(\Delta)$.}
   \label{fig:rank-constraints}
\end{figure}
 

\section{Related Work}
\label{sec:related-work}
\paragraph{Verification.}
Nipkow and Brinkop~\cite{nipkow2019amortized} give the first mechanised amortised analyses of skew heaps, splay trees, splay heaps, and pairing heaps in Isabelle/HOL, but for skew heaps verify only the original bound $3\log(\abs{x}+\abs{y})$, not the improved bound of Kaldewaij and Schoenmakers~\cite{kaldewaij1991derivation}. Lorenz et al.~\cite{leutgeb2021atlas} subsequently verified these data structures using AARA, except for skew heaps. By confirming the improved bound, our work completes the AARA-based verification of all four structures. Unlike the pen-and-paper proof of Kaldewaij and Schoenmakers, our bounds use $\log(\abs{x}+\abs{y}-1)$ instead of $\log(\abs{\texttt{meld}\ x\ y})$, requiring the inequality $\log(\abs{x}+\abs{y}) \geq \log(\abs{\texttt{meld}\ x\ y})$ to be internalised and enabling greater automation.\ifextended\footnote{We also consider alternative bounds, motivating \emph{value variables}, described in Appendix~\ref{ap:value-variables} with an experimental comparison.}\fi

\paragraph{Automated Amortised Resource Analysis.}
Type-based automated amortised analysis originated in constant heap-space bounds~\cite{hofmann2003static} and has since been extended to higher-order functions~\cite{jost2010static}, lazy evaluation~\cite{simoes2012automatic,vasconcelos2015type}, and probabilistic programs~\cite{wang2020raising,leutgeb2022automated}. While non-linear bounds such as polynomial~\cite{hoffmann2010amortized,hoffmann2011multivariate} and exponential~\cite{kahn2020exponential} have been studied, tree-based data structures require logarithmic bounds, first supported by AARA in the analysis of splay trees~\cite{leutgeb2021atlas}, which our work extends.

\section{Conclusion}
\label{sec:conclusion}

In this work we have established a novel type-based amortised cost analysis that can fully automatically infer precise and optimal bounds for sophisticated
functional data structures. The approach has been implemented in a novel prototype implementation and experimentally validated.


Even though we were not able to identify a potential function that resolves the motivating problem of this work—namely, finding a better bound for rank-biased heaps, as posed by Schoenmakers~\cite{schoenmakers2024amortized}—our adaptations confirm existing results in a fully automated fashion and allow us to eliminate a large class of potential functions from the search space.

Since piecewise reasoning and invariants are commonplace in pen-and-paper proofs, we believe that supporting these features enables the analysis of a broad class of data structures. Promising examples include rank-balanced trees~\cite{haeupler2015rank}, such as AVL and red-black trees, as well as maxiphobic heaps~\cite{gibbons2003fun}, which are closely related to weight-biased leftist heaps.
%

\begin{credits}
\subsubsection{\ackname} This study was funded by the Austrian Science Fund (FWF) under Project Number P 36623.

\subsubsection{\discintname} The authors have no competing interests.

\end{credits}

\section{Data-Availability Statement}
A research artifact is publicly available on Zenodo at \url{https://doi.org/10.5281/zenodo.19660281}. It includes a pre-built Docker image of our prototype, provided as a command-line tool, the complete source code, and the benchmark suite used in our evaluation. The artifact also provides detailed, step-by-step instructions for reproducing the results reported in this paper, namely the automatically derived complexity bounds.

In addition, the tool is released as open-source software and is hosted at \url{https://github.com/autosard/atlas-re}, facilitating reuse, inspection, and adaptation for independent studies.

\printbibliography
\ifextended
\newpage
\appendix

\section{Exact Bounds}
\label{sec:exact-bounds}
Table~\ref{tab:bounds} lists the exact cost bounds obtained for our primary case studies, including the concrete coefficients produced by our implementation.

\begin{table}[t]
  \centering
  \caption{Exact cost bounds of \texttt{meld} and \texttt{delete\_min} for skew and leftist heaps.} 
  \label{tab:bounds}
  \begin{tabularx}{\textwidth}{l<{\hspace{3mm}}Xl}
    \toprule
    Potential & \texttt{meld} & \texttt{delete\_min} \\
    \midrule
    \multicolumn{3}{c}{skew heap}\\
    \addlinespace
    $\Phi_{[t < u]}$ & $\log(\abs{y} + \abs{x} - 1) + \log \abs{x} + \log\abs{y}$ &$3 \log \abs{x}$ \\ 
    \addlinespace
    $\Phi_{u}$ &$ \log\abs{x} + \log\abs{y}$ & $\sfrac{3}{2}\log\abs{x}$ \\ 
    \addlinespace
    $\Phi_{-t}$ & $\sfrac{1}{2}\log(\abs{y} + \abs{x} - 1) + \sfrac{1}{2}\log\abs{x} + \sfrac{1}{2}\log\abs{y}$ & $\sfrac{3}{2}\log\abs{x}$ \\ 
    \addlinespace
    $\Phi_{\phi}$ & $\frac{105}{163}\log(\abs{y} + \abs{x} - 1) + \frac{3115}{7824}\log\abs{x} + \frac{3115}{7824}\log\abs{y}$ & $\frac{5635}{3912}\log\abs{x}$ \\
    \midrule
    \multicolumn{3}{c}{weight-biased leftist heap}\\
    \addlinespace
    n.a. & $\log\abs{x} + \log\abs{y}$ & $2\log\abs{x}$ \\ 
    $\Phi_{\phi}$ & $\frac{105}{163}\log(\abs{y} + \abs{x} - 1)  + \frac{3115}{7824}\log\abs{x} + \frac{3115}{7824}\log\abs{y}$ & $\frac{5635}{3912}\log\abs{x}$ \\ 
    \midrule
    \multicolumn{3}{c}{rank-biased leftist heap}\\
    \addlinespace
    n.a. & $\dag x + \dag y$ & $2\log\abs{x}$ \\
    $\dag x$ & $\log(\abs{x} + \abs{y} - 1)$ &  $2\log\abs{x}$ \\
    \bottomrule
  \end{tabularx}
\end{table}



\section{Supplementary Constraint Generating Functions}
\label{sec:constraints}

\subsection{Template Language Agnostic Constraints} As mentioned above the constraint generators for all algorithmic type rules except $\ctrule{Const}$, $\ctrule{Match}$ and $\ctrule{Let}$ yield template language agnostic definitions, that are presented in Figure~\ref{fig:templ-agnostic-constraints}. The following lemma establishes the soundness of these constraint generators, i.e. the constraints imply the obligations of the corresponding declarative type rules.
\begin{lemma}[Soundness of Template Agnostic Constraint Generators]
Let $C$ be set of generated constraints and let $M$ be a satisfying model for $C$. 
Then the following properties hold:
\begin{enumerate}[label=\roman*., leftmargin=*]
  \setlength{\itemsep}{1ex}
  \item If $C = \ctr*{\tlang{}}{\labVar}(Q,Q', \typed{x}{\alpha})$ then $Q^M(\typed{x}{\alpha}) = Q'^M(\typed{x}{\alpha})$.
  \item If $C = \ctr*{\tlang{}}{\labIte}(Q,Q', P, R, \Gamma)$ then $Q^M(\Gamma) = P^M(\Gamma) + R^M(\Gamma)$.
    
  \item If $C = \ctr*{\tlang{}}{\labShift}(k, Q, Q', P, P', \Gamma)$ then
    
    \begin{enumerate}
      \item $Q^M(\Gamma) = P^M(\Gamma) + k$
      \item $Q'^M(\val) = P'^M(\val) + k$  
    \end{enumerate}
  \item If $C = \ctr*{\tlang{}}{\labApp}(k, Q,Q', P, P', R, R', \typed{\vec{x}}{\vec{\alpha}})$ then
    
    \begin{enumerate}[(a)]
    \item $Q^M(\vec{x}) = P^M(\vec{y}) + k \cdot R^M(\vec{y})$
    \item $Q'^M(\val) = P^M(\val) + k \cdot R'^M(\val)$
    \end{enumerate}
  \item If $C = \ctr*{\tlang{}}{\labTick}(a,Q,Q', P')$ then $Q'^M(\val) = P'^M(\val) - a$.
  \item If $C = \ctr*{\tlang{}}{\labWVar}(Q,Q', P, (\Gamma, \typed{x}{\alpha}))$ then $Q^M(\Gamma,\typed{x}{\alpha}) \geq P^M(\Gamma)$.
  \item If $C = \ctr*{\tlang{}}{\labShare}(Q,Q', P, (\Gamma, \typed{z}{\alpha}), (\Gamma, \typed{x}{\alpha}, \typed{y}{\alpha}))$ then $P^M(\Gamma,\typed{x}{\alpha}, \typed{y}{\alpha}) = Q^M(\Gamma, \typed{z}{\alpha})$. 
  \end{enumerate}
\end{lemma}
\begin{proof}
  Each case follows immediately by unfolding the definition of the corresponding constraint-generating function.
\end{proof}

\begin{figure}
  \begin{mathpar}
    \inferrule*{
      q_i = q'_{i[x \mapsto \val]}
    }{
      \ctr*{\tlang{}}{\labVar}(Q,Q', \typed{x}{\alpha}) 
    }
    \and
    \inferrule*{
      q_i = p_i  + r_i
    }{
      \ctr*{\tlang{}}{\labIte}(Q,\_, P, R, \Gamma)
    }
    \and
    \inferrule*{
      p_{1} = q_{1} - k \\  p_{i \neq 1} = q_{i \neq 1}\\\\
      p'_{1} = q'_{1} - k \\  p'_{i \neq 1} = q'_{i \neq 1}\\
      k \geq 0
    }{
      \ctr*{\tlang{}}{\labShift}(k, Q, Q', P, P', \Gamma)
    }
    \and
    \inferrule*{
      \vec{y} = \m{args}(P) = \m{args}(R)\\ k \in \N\\\\
      q_{i[\vec{x} \mapsto \vec{y}]} = p_i  + k \cdot r_i 
      \\
      q'_{i} = p'_i  + k \cdot r'_i
    }{
      \ctr*{\tlang{}}{\labApp}(k, Q,Q', P, P', R, R', \typed{\vec{x}}{\vec{\alpha}})
    }
    \and 
    \inferrule*{
      p'_{1} = q'_{1} + a \\ p'_{i \neq 1} = q'_{i \neq 1}
    }{
      \ctr*{\tlang{}}{\labTick}(a,\_,Q', P')
    }
    \and
    \inferrule*{
      p_{i \notin I_x} = q_{i \notin I_x} \\ q_{i \in I_x} \geq 0
      \\\\
      I_x = \{i \mid i\text{ contains variable } x\}
    }{
      \ctr*{\tlang{}}{\labWVar}(Q,\_, P, (\Gamma, \typed{x}{\alpha}))
    }
    \and
    \inferrule*{
      p_{i[z \mapsto x]} + p_{i[z \mapsto y]} = q_{i}
    }{
      \ctr*{\tlang{}}{\labShare}(Q,\_, P, (\Gamma, \typed{z}{\alpha}), (\Gamma, \typed{x}{\alpha}, \typed{y}{\alpha}))
    }
  \end{mathpar}
  \caption{Template language agnostic constraint generating functions.}
  \label{fig:templ-agnostic-constraints}
\end{figure}

\subsection{Constraint Generators for the Logarithmic Template Language}
\begin{figure}
  \begin{mathpar}
    \inferrule*{
      q(2) = \sum_{a + b = 2} q'_{(a\val,b)}
    }{
      \ctr*{\tlangLog}{\labConst}(\flstc{leaf}, Q, Q')
    }
    \and
    \inferrule*{
      q_{(at,au,c)} = q'_{(a\val,c)}
    }{
      \ctr*{\tlangLog}{\labConst}(\flstc{node}\ t\ a\ u,Q,Q')
    }
    \and
    \inferrule*{
      p_{(\vec{a}\vec{x}, d)} = \sum_{b+c=d} q_{(\vec{a}\vec{x}, bx,c)}\\
      r_{(\vec{a}\vec{x}, bt, bu, c)} = q_{(\vec{a}\vec{x}, bx, c)}
    }{
      \ctr*{\tlangLog}{\labMatch}
      (Q,\_,P,R,(\Gamma,\typed{x}{\Tree}),
      (\Gamma,\typed{t}{\Tree},\typed{u}{\Tree}))
    }
    \and
    \inferrule*[vskip=1ex]{
      p_{(\vec{a}\vec{x},c)} = q_{(\vec{a}\vec{x},c)}\\
      r_{(dx, e)} = p'_{(d\val,e)}\\
      \forall \vec{b} \neq \vec{0}.\left( \ p_{(\vec{b}\vec{y})} = q_{(\vec{b}\vec{y})} \right)\\
      {
        \begin{array}{l}
          \forall \vec{a} \neq \vec{0} \lor c \neq 0, \vec{b} \neq \vec{0}.\\
          \left (q_{(\vec{a}\vec{x}, \vec{b}\vec{y}, c)} = \sum_{d \ne 0\lor e \ne 0} p^{(\vec{b},d,e)}_{(\vec{a}\vec{x},c)} \right )\\    
        \end{array}
      }\\
      {
        \begin{array}{l}
          \forall \vec{b} \neq \vec{0}, d \neq 0 \lor e \neq 0.\\
          \left ( \begin{array}{lr}
            \sum_{(\vec{a},c)} p_{(\vec{a}\vec{x},c)}^{(\vec{b},d,e)} \geq {p'}_{(d\val,e)}^{(\vec{b},d,e)} & r_{(\vec{b}\vec{y},dx,e)} = p'^{(\vec{b},d,e)}_{(d\val,e)}\\
            \forall \vec{a} \neq 0 \lor c \neq 0.\\
            \left( \ p^{(\vec{b},d,e)}_{(\vec{a}\vec{x},c)} \neq 0 \rightarrow p'^{(\vec{b},d,e)}_{(d\val,e)} \leq p^{(\vec{b},d,e)}_{(\vec{a}\vec{x}, c)}\right) \\
          \end{array}\right )  
        \end{array}
      }
    }{
      \ctr*{\tlangLog}{\labLet}
      (Q,\_,P,P',R,\vec{P},\vec{P'},\Gamma,
      (\Delta,\typed{x}{\Tree}))
    }
  \end{mathpar}
  \caption{Constraint generating functions for $\tlangLog$, where $\vec{x} = \dom(\Gamma)$ and $\vec{y} = \dom(\Delta)$}
  ~\label{fig:log-constraints}
\end{figure}

Instead of directly listing the constraint generators for $\tlangSol$, we first present the constraints for $\tlangLog$ (Figure~\ref{fig:log-constraints}). This allows us to isolate constraints that govern the logarithmic terms of resource templates and to reuse them across the template languages $\tlangSol$, $\tlangPw$, and $\tlangRank$, all of which include logarithmic terms.

Before we can establish the soundness of these constraint generators, we require the following lifting lemma for logarithmic terms. 
\begin{lemma}[Hofmann et al.~\cite{hofmann2022type}]
  \label{lem:log_cf}
  Assume $\sum_i q_i \log a_i \geq q \log b$ for some rational numbers $a_i, b > 0$ and $q_i \geq q$. Then, $\sum_i q_i \log(a_i + c) \geq q \log(b + c)$ for all $c \geq 1$.
\end{lemma}

\begin{lemma}[Soundness of Constraint Generators for $\tlangLog$] Let $C$ be a set of generated constraints and let $M$ be a satisfying model for $C$. 
  Then the following properties hold:
  \label{lem:soundness-tlanglog}
\begin{enumerate}[label=\roman*., leftmargin=*]
  \setlength{\itemsep}{1ex}
\item If $\ctr*{\tlangLog}{\labConst}(\leaf, Q ,Q')$ then $Q^M() = Q'^M(\leaf)$
\item If $\ctr*{\tlangLog}{\labConst}(\tree{t}{a}{u},Q,Q')$ then $Q^M(t,u) = Q'^M(\tree{t}{a}{u})$
\item If $C = \ctr*{\tlangLog}{\labMatch}(Q,Q',P,R,(\Gamma,\typed{x}{\Tree}),(\Gamma,\typed{t}{\Tree},\typed{u}{\Tree}))$ then
  \begin{enumerate}
  \item $P^M(\Gamma) = Q^M(\Gamma, \leaf)$,
  \item $R^M(\Gamma, \typed{t}{\Tree}, \typed{u}{\Tree})) = Q^M(\Gamma, \tree{t}{a}{u})$.  
  \end{enumerate}
\item If $C = \ctr*{\tlangLog}{\labLet}(Q,Q',P,P',R,\vec{P},\vec{P'},\Gamma,(\Delta,\typed{x}{\Tree}))$ then
  \begin{enumerate}
  \item $Q^M(\Gamma, \Delta) = P^M(\Gamma) + \Psi_{\Delta}(\Delta) + \Psi_{\text{mix}}(\Gamma,\Delta)$
  \item $R^M(\typed{x}{\alpha},\Delta) = P'^M(e_1) + \Psi_{\Delta}(\Delta) + \Omega(\typed{x}{\alpha},\Delta)$
  \item $(P^{i})^M(\Gamma) \geq (P'^{i})^M(\typed{x}{\alpha}) \Rightarrow \Psi_{\text{mix}}(\Gamma,\Delta) \geq \Omega(\typed{x}{\alpha},\Delta)$
  \end{enumerate}
\end{enumerate}

\end{lemma}

\begin{proof}
  Claims (i)-(iii) follow directly from inspection of the given constraints. For (iv) we begin by verifying that $Q(\Gamma, \Delta) = P(\Gamma) + \Psi_{\Delta}(\Delta) + \Psi_{\text{mix}}(\Gamma,\Delta)$ follows from the given template constraints. For this we choose $\Psi_{\text{mix}}(\Gamma,\Delta) = \sum_{\vec{b} \neq 0, \vec{a} \neq 0 \lor c \neq 0} q_{(\vec{a}\vec{x},\vec{b}\vec{y},c)} \log(\vec{a}\abs{\vec{x}} + \vec{b}\abs{\vec{y}} + c)$.
  \begin{align*}
    Q(\Gamma, \Delta) & = \sum_{\vec{a}, c} q_{(\vec{a}\vec{x}, c)} \log(\vec{a}\abs{\vec{x}} + c)
            + \sum_{\vec{b} \neq \vec{0}} q_{(\vec{b}\vec{y})}  \log(\vec{b} \abs{\vec{y}}) + \Psi_{\text{mix}}(\Gamma,\Delta)\\
            & = \sum_{\vec{a}, c} p_{(\vec{a}\vec{x}, c)} \log(\vec{a}\abs{\vec{x}} + c) + \sum_{\vec{b} \neq \vec{0}} q_{(\vec{b}\vec{y})}  \log(\vec{b} \abs{\vec{y}}) + \Psi_{\text{mix}}(\Gamma,\Delta)\\
            & = P(\Gamma) + \Psi_{\Delta}(\Delta) + \Psi_{\text{mix}}(\Gamma,\Delta)
  \end{align*}
  We proceed in the same way to verify $R(\typed{x}{\alpha},\Delta) = P'(e_1) + \Psi_{\Delta}(\Delta) + \Omega(\typed{x}{\alpha},\Delta)$, where we chose $\Omega(\typed{x}{\alpha},\Delta) = \sum_{\vec{b} \neq \vec{0}, d \neq 0 \lor e \neq 0}  r'_{(\vec{b}\vec{y},dx,e)}\log(\vec{b}\abs{\vec{y}} + dx + e)$.
  \begin{align*}
    R(\typed{x}{\alpha},\Delta) & = \sum_{d,e} r_{(dx,e)} \log(dx + e) 
                        + \sum_{\vec{b} \neq \vec{0}} r_{(\vec{b}\vec{y})} \log(\vec{b} \abs{\vec{y}}) 
                         + \Omega(\typed{x}{\alpha},\Delta)\\
                        & = \sum_{d,e} p'_{(d\val,e)} \log(d\val + e) 
                        + \sum_{\vec{b} \neq \vec{0}} q_{(\vec{b}\vec{y})}  \log(\vec{b} \abs{\vec{y}}) + \Omega(\typed{x}{\alpha},\Delta) \\
                        & = P'(e_1) + \Psi_{\Delta}(\Delta) + \Omega(\typed{x}{\alpha},\Delta)
  \end{align*}
  Finally we verify the constraints concerning the lifting of mixed variable potential, i.e. $\bigl( \Psi^{(\vec{b},d,e)}(\Gamma) \geq \Phi^{(\vec{b},d,e)}(\typed{x}{\alpha})\bigr)^{\ddag} \Rightarrow \Psi_{\text{mix}}(\Gamma,\Delta) \geq \Omega(\typed{x}{\alpha},\Delta)$. In Line 4 we use the template constraints together with the premise $(\ddag)$, to instantiate Lemma~\ref{lem:log_cf}. 
  \begin{align*}
    \Psi_{\text{mix}}(\Gamma,\Delta) & = \sum_{\vec{b} \neq 0, \vec{a} \neq 0 \lor c \neq 0} q_{(\vec{a}\vec{x},\vec{b}\vec{y},c)} \log(\vec{a}\abs{\vec{x}} + \vec{b}\abs{\vec{y}} + c)\\
                      & = \sum_{\vec{b} \neq 0, \vec{a} \neq 0 \lor c \neq 0} \sum_{d \ne 0\lor e \ne 0} p^{(\vec{b},d,e)}_{(\vec{a}\vec{x},c)} \log(\vec{a}\abs{\vec{x}} + \vec{b}\abs{\vec{y}} + c)\\
                        & = \sum_{\vec{b} \neq \vec{0}, d \neq 0 \lor e \neq 0}\sum_{\vec{a} \neq \vec{0} \lor c \neq 0}  p^{(\vec{b},d,e)}_{(\vec{a}\vec{x},c)} \log(\vec{a}\abs{\vec{x}} + \vec{b}\abs{\vec{y}} + c)\\
                        & \geq \sum_{\vec{b} \neq \vec{0}, d \neq 0 \lor e \neq 0} p'^{(\vec{b},d,e)}_{(dx,e)}\log(\vec{b}\abs{\vec{y}} + dx + e) & \text{Lemma~\ref{lem:log_cf}}, (\ddag)\\
                        & = \sum_{\vec{b} \neq \vec{0}, d \neq 0 \lor e \neq 0} r'_{(\vec{b}\vec{y},dx,e)}\log(\vec{b}\abs{\vec{y}} + dx + e) \\
                        & = \Omega(\typed{x}{\alpha},\Delta)
  \end{align*}
\end{proof}  

\subsection{Constraint Generators for the Sum-Of-Logs Template Language}

In Figure 10 we present the complete constraint generators for $\tlangLog$. To support a more compact presentation when including the constraints for $\tlangLog$ in the generators for $\tlangSol$, we use the convention that a specific coefficient is never constrained more than once, i.e. the definition of $\tlangSol$ overrides $\tlangLog$, wherever there are conflicting constraints. This is relevant for example when constraining $r_{(t,u)}$ for $\rulematch$ in the $\node$ case, which is already constrained in the premises. 

\begin{figure}[t]
  \begin{mathpar}
    \inferrule*{
      q(2) = \sum_{a + b = 2} q'_{(a\val,b)} + q'_{\val}\\\\
      \ctr*{\tlangLog}{\labConst}(\flstc{leaf},Q,Q')
    }{
      \ctr*{\tlangSol}{\labConst}(\flstc{leaf},Q,Q')
    }
    \and
    \inferrule*{
      q_{(t)} = d \cdot q'_{\val}\\
      q_{(u)} = e \cdot q'_{\val}\\\\
      q_{(t,u)} = f \cdot q'_{\val} + q'_{(\val)}\\
      \ctr*{\tlangLog}{\labConst}(\flstc{node}\ t\ a\ u,Q,Q')
    }{
      \ctr*{\mathcal{L}_{\mathrm{sol},(d,e,f)}}{\labConst}(\flstc{node}\ t\ a\ u,Q,Q')
    }
    \and
    \inferrule*{
      p_{(\vec{a}\vec{x}, 2)} = \sum_{a + b = 2} q_{(\vec{a}\vec{x},ax,b)} + q_x[\vec{a} = \vec{0}]\\
      \ctr*{\tlangLog}{\labMatch}
      (Q,Q',P,R,(\Gamma,\typed{x}{\Tree}),
      (\Gamma,\typed{t}{\Tree},\typed{u}{\Tree}))\\
      p_{x_i} = r_{x_i} = q_{x_i}\\
      r_{(t)} = d \cdot q_x\\
      r_{(u)} = e \cdot q_x\\
      r_{(t,u)} = f \cdot q_{x} + q_{(t,u)}
    }{
      \ctr*{\mathcal{L}_{\mathrm{sol},(d,e,f)}}{\labMatch}
      (Q,Q',P,R,(\Gamma,\typed{x}{\Tree}),
      (\Gamma,\typed{t}{\Tree},\typed{u}{\Tree}))
    }
    \and
    \inferrule*{
      p_{x_i} = q_{x_i}\\
      r_{y_i} = q_{y_i}\\
      r_{x} = p'_{\val}\\\\
      \ctr*{\tlangLog}{\labLet}
      (Q,Q',P,P',R,\vec{P},\vec{P'},\Gamma,
      (\Delta,\typed{x}{\Tree}))\\
    }{
      \ctr*{\tlangSol}{\labLet}
      (Q,Q',P,P',R,\vec{P},\vec{P'},\Gamma,
      (\Delta,\typed{x}{\Tree}))
    }
  \end{mathpar}
  \caption{Constraint generating functions for $\tlangSol$, where $\vec{x} = \dom(\Gamma)$ and $x_i \in \dom(\Gamma)$. $[p]$ denotes the Iverson bracket for predicate $p$, i.e. $[p] = 1$ if $p$ holds and $[p] = 0$.}
  \label{fig:sol-constraints}
\end{figure}

\begin{lemma}[Soundness of Constraint Generators for $\tlangSol$] Let $C$ be set of generated constraints and let $M$ be a satisfying model for $C$. 
Then the following properties hold:
\begin{enumerate}[label=\roman*., leftmargin=*]
  \setlength{\itemsep}{1ex}
\item If $\ctr*{\tlangSol}{\labConst}(\leaf, Q,Q')$ then $Q^M() = Q'^M(\leaf)$
\item If $\ctr*{\tlangSol}{\labConst}(\tree{t}{a}{u},Q,Q')$ then $Q^M(t,u) = Q'^M(\tree{t}{a}{u})$
\item If $C = \ctr*{\tlangSol}{\labMatch}(Q,Q',P,R,(\Gamma,\typed{x}{\Tree}),(\Gamma,\typed{t}{\Tree},\typed{u}{\Tree}))$ then
  \begin{enumerate}
  \item $P^M(\Gamma) = Q^M(\Gamma, \leaf)$,
  \item $R^M(\Gamma, \typed{t}{\Tree}, \typed{u}{\Tree})) = Q^M(\Gamma, \tree{t}{a}{u})$.  
  \end{enumerate}
\item If $C = \ctr*{\tlangSol}{\labLet}(Q,Q',P,P',R,\vec{P},\vec{P'},\Gamma,(\Delta,\typed{x}{\Tree}))$ then
  \begin{enumerate}
  \item $Q^M(\Gamma, \Delta) = P^M(\Gamma) + \Psi_{\Delta}(\Delta) + \Psi_{\text{mix}}(\Gamma,\Delta)$
  \item $R^M(\typed{x}{\alpha},\Delta) = P'^M(e_1) + \Psi_{\Delta}(\Delta) + \Omega(\typed{x}{\alpha},\Delta)$
  \item $(P^{i})^M(\Gamma) \geq (P'^{i})^M(\typed{x}{\alpha}) \Rightarrow \Psi_{\text{mix}}(\Gamma,\Delta) \geq \Omega(\typed{x}{\alpha},\Delta)$
  \end{enumerate}
\end{enumerate}  
\end{lemma}
\begin{proof}
  We take the results for templates with logarithmic terms established in Lemma~\ref{lem:soundness-tlanglog} and verify that the properties also hold for the additional terms of $\tlangLog$, by inspecting the given constraints.  
\end{proof}
\subsection{Constraint Generators for Piecewise Template Language}

The constraint generators for $\tlangPw$ are shown in Figure~\ref{fig:piece-constraints} in Section~\ref{sec:case-studies}. The following lemma established the soundness of these constraints.
\begin{lemma}[Soundness of Constraint Generators for $\tlangPw$] Let $C$ be set of generated constraints and let $M$ be a satisfying model for $C$. 
Then the following properties hold:
\begin{enumerate}[label=\roman*., leftmargin=*]
  \setlength{\itemsep}{1ex}
\item If $\ctr*{\tlangPw}{\labConst}(\leaf, Q,Q')$ then $Q^M() = Q'^M(\leaf)$
\item If $\ctr*{\tlangPw}{\labConst}(\tree{t}{a}{u},Q,Q')$ then $Q^M(t,u) = Q'^M(\tree{t}{a}{u})$
\item If $C = \ctr*{\tlangPw}{\labMatch}(Q,Q',P,R,(\Gamma,\typed{x}{\Tree}),(\Gamma,\typed{t}{\Tree},\typed{u}{\Tree}))$ then
  \begin{enumerate}
  \item $P^M(\Gamma) = Q^M(\Gamma, \leaf)$,
  \item $R^M(\Gamma, \typed{t}{\Tree}, \typed{u}{\Tree})) = Q^M(\Gamma, \tree{t}{a}{u})$.  
  \end{enumerate}
\item If $C = \ctr*{\tlangPw}{\labLet}(Q,Q',P,P',R,\vec{P},\vec{P'},\Gamma,(\Delta,\typed{x}{\Tree}))$ then
  \begin{enumerate}
  \item $Q^M(\Gamma, \Delta) = P^M(\Gamma) + \Psi_{\Delta}(\Delta) + \Psi_{\text{mix}}(\Gamma,\Delta)$
  \item $R^M(\typed{x}{\alpha},\Delta) = P'^M(e_1) + \Psi_{\Delta}(\Delta) + \Omega(\typed{x}{\alpha},\Delta)$
  \item $(P^{i})^M(\Gamma) \geq (P'^{i})^M(\typed{x}{\alpha}) \Rightarrow \Psi_{\text{mix}}(\Gamma,\Delta) \geq \Omega(\typed{x}{\alpha},\Delta)$
  \end{enumerate}
\end{enumerate}
\end{lemma}
\begin{proof}
  Similar to above we obtain the Lemma by inspection of the given constraints. 
\end{proof}
\subsection{Constraint Generators for Rank Template Language}
The constraint generators for $\tlangRank$ are provided in Figure~\ref{fig:rank-constraints} in Section~\ref{sec:case-studies} and their soundness is established in the following Lemma.
\begin{lemma}[Soundness of Constraint Generators for $\tlangRank$] Let $C$ be set of generated constraints and let $M$ be a satisfying model for $C$. 
Then the following properties hold:
\begin{enumerate}[label=\roman*., leftmargin=*]
  \setlength{\itemsep}{1ex}
\item If $\ctr*{\tlangRank}{\labConst}(\leaf, Q,Q')$ then $Q^M() = Q'^M(\leaf)$
\item If $\ctr*{\tlangRank}{\labConst}(\tree{t}{a}{u},Q,Q')$ then $Q^M(t,u) = Q'^M(\tree{t}{a}{u})$
\item If $C = \ctr*{\tlangRank}{\labMatch}(Q,Q',P,R,(\Gamma,\typed{x}{\Tree}),(\Gamma,\typed{t}{\Tree},\typed{u}{\Tree})$ then
  \begin{enumerate}
  \item $P^M(\Gamma) = Q^M(\Gamma, \leaf)$,
  \item $R^M(\Gamma, \typed{t}{\Tree}, \typed{u}{\Tree})) = Q^M(\Gamma, \tree{t}{a}{u})$.  
  \end{enumerate}
\item If $C = \ctr*{\tlangRank}{\labLet}(Q,Q',P,P',R,\vec{P},\vec{P'},\Gamma,(\Delta,\typed{x}{\Tree}))$ then
  \begin{enumerate}
  \item $Q^M(\Gamma, \Delta) = P^M(\Gamma) + \Psi_{\Delta}(\Delta) + \Psi_{\text{mix}}(\Gamma,\Delta)$
  \item $R^M(\typed{x}{\alpha},\Delta) = P'^M(e_1) + \Psi_{\Delta}(\Delta) + \Omega(\typed{x}{\alpha},\Delta)$
  \item $(P^{i})^M(\Gamma) \geq (P'^{i})^M(\typed{x}{\alpha}) \Rightarrow \Psi_{\text{mix}}(\Gamma,\Delta) \geq \Omega(\typed{x}{\alpha},\Delta)$
  \end{enumerate}
\end{enumerate}
\end{lemma}
\begin{proof}
  Again, the lemma follows straightforwardly from the definition of the constraint generators.
\end{proof}

\section{Supplementary Type Derivation Examples}
\label{sec:additional-case-studies}

In this section we illustrate how the previously introduced template languages and constraint systems are employed in concrete analyses. We motivate the design choices for the template languages and expert knowledge by highlighting the key steps in the type derivations for each data structure. This makes explicit how the structure of the derivations determines the required forms of potential and expert knowledge.

\subsection{Skew Heaps}

\paragraph{Weakening.} Early in the type derivation of \texttt{meld x y} with $\tlangPw$, we encounter a representative application of the weakening rule~\(\rulew\). Concretely, for \(x = \tree{t}{a}{u}\), the application of \(\rulew\) corresponds to the following inequality: 
\begin{align*}
  &\underbrace{\log(\abs{t} + \abs{u}) + \log\abs{y} + \log(\abs{t} + \abs{u} + \abs{y} - 1)}_{\mathcal{A}_{\texttt{meld}}(x,y)} + [\abs{t} < \abs{u}]\\
  \geq &\quad \underbrace{\log\abs{u} + \log\abs{y} + \log(\abs{u} + \abs{y} - 1) + 1}_{\mathcal{A}_{\texttt{meld}}(u,y)} +  [\abs{t} > \abs{u} + \abs{y} - 1]
\end{align*}
It relates the amortised costs of \texttt{meld} with the amortised costs of its recursive call. The additional term $[\abs{t} > \abs{u} + \abs{y} - 1]$ is crucial to provide the potential for the result tree. Essentially we are exploiting two inequalities here. We have $\log(\abs{t} + \abs{u}) + [\abs{t} < \abs{u}] \geq \log\abs{u} + 1$ and $\log(\abs{t} + \abs{u} + \abs{y} - 1) \geq \log(\abs{u} + \abs{y} - 1) + [\abs{t} > \abs{u} + \abs{y} - 1]$. Both can be derived from Lemma~\ref{lem:logLemma2Iver}.

\paragraph{Lifting of Iverson Terms.} The next interesting step lifts the potential $[\abs{t} > \abs{u} + \abs{y} - 1]$ through the let-bound recursive call $\texttt{meld}\ u\ y$ to the binding variable $z$, in order to pay $[\abs{t} > \abs{z}]$ to construct the result node. We remember that $\rulelet$ has the obligation $\Psi_{\text{mix}}(\Gamma,\Delta) \geq \Omega(\typed{x}{\alpha},\Delta)$, concerning terms that mix variables from the binding ($\Gamma$), the body ($\Delta$) and the binding variable $x$. In this case we need to ensure that $[\abs{t} > \abs{u} + \abs{y} - 1] \geq [\abs{t} > \abs{z}]$. 
This follows from the inequality $[\abs{u} + \abs{y} - 1 < \abs{X}] \geq [\abs{z} < \abs{X}]$, which is established by the following cost-free derivation:
  \begin{equation*}
    \small
  \inferrule*[right=\ruleapp]{
    \atypdcl{\typed{x}{\Tree},\typed{y}{\Tree}}{[\abs{x} + \abs{y} - 1 < \abs{X}]}{\Tree}{[\abs{\val} < \abs{X}]} \in \mathcal{F}^{\text{cf}}(\text{\lstinline{meld}})
  }{
    \tjudgenacf{\typed{u}{\Tree},\typed{y}{\Tree}\mid[\abs{u} + \abs{y} - 1 < \abs{X}]}{\text{\lstinline{meld}}\ y\ u}{\Tree}{[\abs{\val} < \abs{X}]}
  }
\end{equation*}
By instantiating the template variable $X$ with $t$, we obtain $[\abs{u} + \abs{y} - 1 < \abs{t}] \geq [\abs{z} < \abs{t}]$. The derivation of the cost-free signature itself involves the same application to $\ruleapp$ but here we obtain $[\abs{t} + \abs{u} + \abs{y} - 1 < \abs{X}] \geq [\abs{z} + \abs{t} < \abs{X}]$ by adding $\abs{t}$ to both sides.

\paragraph{Negative Potential.} Setting $\params = (-\sfrac{1}{2},0,\sfrac{1}{2})$ for $\tlangSol$, as we suggest above, crucially creates a negative term in the potential function. This negative term allows to ``borrow'' potential. This means we can delay the payment of potential, temporarily taking on negative potential and since we restrict the potential returned by the function to be non-negative, we guarantee that potential is payed back eventually.

Concretely we can witness this pattern in the application of the rule $\rulelet$ shown in Figure~\ref{fig:let-borrowing}, which covers the let-bound recursive call to $\texttt{meld}\ u\ y$:
\begin{figure}[t]
  \centering
\begin{equation*}
  \scriptsize
    \inferrule*[lab=\rulelet]{
      \inferrule*[vskip=1ex,right=\ruletick]{
        \inferrule*[right=\ruleshift,rightskip=3.1em]{
          \inferrule*[right=\ruleapp]{
            \dots \quad \atypdcl{\typed{x}{\Tree}, \typed{y}{\Tree}}{-\sfrac{1}{2}\log(\abs{x} + \abs{y} - 1)}{\Tree}{-\sfrac{1}{2}\log\abs{\val}} \in \FScf(\texttt{meld})\\
          }{
            \tjudgeml{\typed{u}{\Tree},\typed{y}{\Tree}\mid \sfrac{1}{2}\log\abs{y}  + \sfrac{1}{2}\log\abs{u} + \phiSol(y) + \phiSol(u)}{\text{\lstinline{meld}}\ y\ u}{\Tree}{\phiSol(\val) - \sfrac{1}{2}\log\abs{\val}}
          }
        }{ }
      }{
        \tjudgeml{\typed{u}{\Tree},\typed{y}{\Tree}\mid \sfrac{1}{2}\log\abs{y}  + \sfrac{1}{2}\log\abs{u} + \phiSol(y) + \phiSol(u) + 1}{\text{\lstinline{\~ meld}}\ y\ u}{\Tree}{\phiSol(\val) - \sfrac{1}{2}\log\abs{\val}}
      }
        \\
        \tjudge{\typed{z}{\Tree},\typed{t}{\Tree}}{\phiSol(z) - \sfrac{1}{2}\log\abs{z} + \dots}{e}{\Tree}{\phiSol(\val)}
      }{
        \tjudgeml{\typed{t}{\Tree}, \typed{u}{\Tree},\typed{y}{\Tree}\mid \sfrac{1}{2}\log\abs{u} + \sfrac{1}{2}\log\abs{y} + 1 + \phiSol(u) + \phiSol(y) + \dots}{\vlet \ z = \text{\lstinline{meld}}\ y\ u\ \vin\ \dots }{\Tree}{\phiSol(\val)}
      }
    \end{equation*}
    \caption{Partial type derivation of \texttt{meld}, showcasing potential ``borrowing''. }
    \label{fig:let-borrowing}
    \end{figure}
The recursive call requires a potential of $\sfrac{1}{2}\log(\abs{x} + \abs{y} - 1) + \sfrac{1}{2}\log\abs{x} + \sfrac{1}{2}\log\abs{y} + \phiSol(x) + \phiSol(y)$, according to the costed signature of \texttt{meld}. Our analysis additionally derives the shown cost-free signature for \texttt{meld}, which is valid since $\abs{\mathtt{meld}\ x\ y} = \abs{x} + \abs{y} - 1$. The rule $\ruleapp$ then combines the two signatures. We are effectively ``borrowing'' the potential $\log\abs{\val}$ from the result of the call to pay the cost $\sfrac{1}{2}\log(\abs{x} + \abs{y} - 1)$.

\subsection{Weight-Biased Leftist Heaps}
\paragraph{Path Sensitivity.} The step that really distinguishes the type derivation for weight-biased leftist heaps from skew heaps is the typing of the balancing function. The following derivation illustrates the branch in which the sub-trees are not rotated. Whereas in the rotation-case we can use the potential released by the rotation to construct the final tree, in this case we instead rely on the predicate $\#{t} > \#{u}$, added to the type context, according to the branch condition. 
\begin{equation*}
  \small
    \inferrule*[right=\ruleite]{
    \dots
    \\
    \inferrule*[right=\rulew]{
      \inferrule*[right=\ruleconst]{ }{
        \tjudgeml{\typed{t}{\WeightTree}, \typed{u}{\WeightTree};\#{t} > \#{u} \mid \sfrac{1}{2}\log(\abs{t} + \abs{u}) - \sfrac{1}{2}\log\abs{t} \\
          + \phiSol(t) + \phiSol(u) }{\tree{u}{a}{t}}{\WeightTree}{\phiSol(\val)}
      }
    }{
      \tjudgeml{\typed{t}{\WeightTree}, \typed{u}{\WeightTree};\#{t} > \#{u} \mid \sfrac{1}{2}\log(\abs{t} + \abs{u}) - \sfrac{1}{2}\log\abs{u} \\
        + \phiSol(t) + \phiSol(u) }{\tree{t}{a}{u}}{\WeightTree}{\phiSol(\val)}
    }
  }{
    \tjudgeml{\typed{t}{\WeightTree}, \typed{u}{\WeightTree} \mid \sfrac{1}{2}\log(\abs{t} + \abs{u}) - \sfrac{1}{2}\log\abs{u} \\
      + \phiSol(t) + \phiSol(u) }{\cif\ t \leq u\ \cthen\ \tree{u}{a}{t}\ \celse\ \tree{t}{a}{u}}{\WeightTree}{\phiSol(\val)}
  }
\end{equation*}
The predicate implies $\abs{t} \geq \abs{u}$, which in turn allows the rule $\rulew$ to weaken $-\log\abs{u}$ to $-\log\abs{t}$ as required by the definition of $\phiSol$.

\paragraph{Worst-Case Costs} To confirm the known worst case costs for this data structure, we simply set the potential function to zero, i.e. establishing $\tjudge{\Gamma; \mathcal{G}}{\Psi}{\texttt{meld}\ x\ y}{\alpha}{0}$. To obtain logarithmic worst-case bounds, we must, unsurprisingly, take the corresponding data structure invariants into account. In particular the invariant $\#{t} \geq \#{u}$, is required to pay for the recursive call in \texttt{meld}, i.e., we obtain $\log(\abs{t} + \abs{u}) \geq \log\abs{u} + 1$ by Corollary~\ref{lem:guarded-favorite-log}.

\section{Soundness Proof}
\label{sec:soundness-proof}
\soundnessTheorem
\begin{proof}
The proof is conducted via main induction on $\Pi: \eval{\sigma}{c}{e}{v}$ and side induction on $\Xi: \tjudge{\Gamma; \mathcal{G}}{\Psi}{e}{\alpha}{\Phi}$. In most cases, the guard predicates $\bigwedge \mathcal{G}$ are not required; hence we omit the implication from the intermediate statements, as it can always be reintroduced when needed. We begin by covering all cases in which $\Xi$ ends with a structural rule and assume in the further cases that it ends with syntax directed rule. 
  \begin{description}
    \setlength\itemsep{0.5cm}
  \item[Case.] $\Pi$ derives $\eval{\sigma}{c}{e}{v}$ and $\Xi$ ends with a application of the rule $\ruleshift$:
    \begin{equation*}
      \inferrule*[right=\rlabel{\ruleshift}]{%
  \tjudge{\Gamma}{\potAltSymb}{e}{\alpha}{\potSymb}\\
  k \geqslant 0
}{
  \tjudge{\Gamma}{\potAltSymb + k}{e}{\alpha}{\potSymb + k}
}

    \end{equation*}
    From the MIH we obtain $\Psi(\Gamma)\sigma \geq c +  \Phi(v)$. The theorem then is trivially shown as follows.
    \begin{align*}
      (\Psi + k)(\Gamma\sigma) &= \Psi(\Gamma\sigma) + k\\
                  & \geq c + \Phi(v) + k & \text{MIH}\\
                  & \geq c + (\Phi + k)(v)
    \end{align*}
  \item[Case.] $\Pi$ derives $\eval{\sigma}{c}{e}{v}$ and $\Xi$ ends with a application of the rule $\rulew$:
    \begin{equation*}
      \inferrule*[right=\rlabel{\rulew}]{
  \tjudge{\Gamma\changed{; \mathcal{G}}}{\potAltSymb'}{e}{\alpha}{\potSymb'}
  \\
  \changed{\bigwedge \mathcal{G} \Rightarrow} \potArg{\potAltSymb'}{\Gamma} \leqslant \potAlt{\Gamma}
  \\
  \potArg{\potSymb'}{\typed{x}{\alpha}} \geqslant \pot{\typed{x}{\alpha}}
}{
  \tjudge{\Gamma\changed{; \mathcal{G}}}{\potAltSymb}{e}{\alpha}{\potSymb}
}

    \end{equation*}
    We conclude the case as follows.
    \begin{align*}
      \bigwedge \mathcal{G} \Rightarrow \Psi(\Gamma\sigma) &\geq \Psi'(\Gamma\sigma) &\text{Premise II}\\
                  & \geq c +  \Phi'(v) & \text{MIH}\\
                  & \geq c +  \Phi(v) & \text{Premise III}
    \end{align*}
  \item[Case.] $\Pi$ derives $\eval{\sigma}{c}{e}{v}$ and $\Xi$ ends with a application of the rule $\ruleshare$:
    \begin{equation*}
      \inferrule*[right=\rlabel{\ruleshare}]{
  \Psi'(\Gamma,\typed{x}{\alpha},\typed{y}{\alpha}) = \Psi(\Gamma,\typed{z}{\alpha})
  \\
  \tjudge{\Gamma, \typed{x}{\alpha}, \typed{y}{\alpha}}{\Psi'}{e[x,y]}{\beta}{\Phi}
}{
  \tjudge{\Gamma, \typed{z}{\alpha}}{\potAltSymb}{e[z,z]}{\beta}{\potSymb}
}

    \end{equation*}
    We then have
    \begin{align*}
      \Psi((\Gamma, \typed{z}{\alpha})\sigma) & = \Psi'((\Gamma, \typed{x}{\alpha}, \typed{y}{\alpha})\sigma) & \text{Premise I}\\
                          & \ge c + \Phi(v) & \text{MIH}
    \end{align*}

  \item[Case.] Consider $\Xi$ ending in the final structural rule $\rulewvar$ and let $\Pi$ derive $\eval{\sigma}{c}{e}{v}$.
    \begin{equation*}
      \inferrule*[right=\rlabel{\rulewvar}]{
  \changed{\Psi(\Gamma, \typed{x}{\alpha}) \geq \Psi'(\Gamma)}
  \\\\
  \tjudge{\Gamma}{\potAltSymb'}{e}{\beta}{\potSymb}
}{
  \tjudge{\Gamma, \typed{x}{\alpha}}{\potAltSymb}{e}{\beta}{\potSymb}
}

    \end{equation*}
    From the premises we, we obtain the theorem as follows. 
    \begin{align*}
      \Psi(\Gamma, \typed{x}{\alpha})\sigma &\geq \Psi(\Gamma\sigma) & \text{Premise II}\\
                          &= \Psi'(\Gamma\sigma) & \text{Premise I}\\
                          & \ge c + \Phi'(v) & \text{MIH}
    \end{align*}
  \item[Case.] $\Pi$ derives $\eval{\sigma}{0}{x}{v}$ and we have $x\sigma = v$. Then $\Xi$ consists of a single application of the rule $\rulevar$:
    \begin{equation*}
      \inferrule*[right=\rlabel{\rulevar}]{
}{
  \tjudge{\typed{x}{\alpha}}{\Phi}{x}{\alpha}{\Phi}
}

    \end{equation*}
    By applying $\sigma$ we trivially obtain the theorem, $\Phi((\typed{x}{\alpha})\sigma) = \Phi(v)$.
  \item[Case.] $\Pi$ derives $\eval{\sigma}{0}{\const x_1\,\dots\,x_n}{\const v_1\,\dots\,v_n}$ where $\sigma$ satisfies $x_1\sigma = v_1, \dots, x_n\sigma = v_n$. In this case $\Xi$ consists of a single application of the rule $\ruleconst$ or an application of the rule $\ruleunfold$.
    \begin{description}
        \item[Case.] The rule $\ruleconst$ was applied.
          \begin{equation*}
            \inferrule*[right=\rlabel{\ruleconst}]{
 \changed{p_{\alpha,\const}(\vec{x}) \Rightarrow} \ \Psi(\vec{x}) = \Phi(\typed{\const\ \vec{x}}{\alpha})}
{
  \tjudge{\typed{\vec{x}}{\vec{\alpha}}\changed{;p_{\alpha,\const}(\vec{x})}}{\Psi}{\const\ \vec{x}}{\alpha}{\Phi}
}

          \end{equation*}
          We instantiate the premise with the substitution $\sigma$ to obtain the theorem:
          \begin{align*}
            p_{\alpha,\const}(\vec{x}) \Rightarrow \Psi(\vec{x}\sigma) = \Phi(\typed{\const\ \vec{x}\sigma}{\alpha}) = \Phi({\const v_1\,\dots\,v_n}) = \Phi(v)
          \end{align*}
        \item[Case.] The rule $\ruleunfold$ was applied.
          \begin{equation*}
            \inferrule*[Right=\rlabel{\ruleunfold}]{
  \Psi'(\vec{x}) = \Psi(\vec{x}, \typed{\val}{\alpha})\\
  \tjudge{\vec{x}}{\Psi'}{\const\ \vec{x}}{\alpha}{\Phi}
}{
  \tjudge{\vec{x},\typed{\val}{\alpha}}{\Psi}{\const\ \vec{x}}{\alpha}{\Phi}
}

          \end{equation*}
          With an application of the MIH we obtain:
          \begin{align*}
            \Psi((\vec{x}, \typed{\val}{\alpha})\sigma) & = \Psi'(\vec{x}\sigma) & \text{Premise I}\\
            & \ge \Phi(v) & \text{MIH}\\
          \end{align*}
    \end{description}
  \item[Case.] Assume $\Pi$ ends as follows
    \begin{equation*}
      \inferrule*[right=\rlabel{\erule{Match}}]{
  x\sigma = \const v_1\,\dots\,v_n \\ \eval{\sigma''}{c}{e}{v}}
  {\eval{\sigma}{c}{\flstk{match }x\flstk{ with} \begin{array}[t]{l}%
    \ \dots \quad \flst{| } (\const x_1\,\dots\,x_n) \flst{ -> } e \quad \dots
  \end{array}}{v}}

    \end{equation*}
    Wlog.\ we may assume that $\Xi$ ends in an application of the rule $\rulematch$:
    \begin{equation*}
      \inferrule*[right=\rlabel{\rulematch}]{
\changed{p_{\alpha,\const_i}(\vec{x}_i) \Rightarrow} \Psi_i(\Gamma, \vec{x}_i) = \potAlt{\Gamma, \const_i \vec{x}_i} \\
 \tjudge{\Gamma, \typed{\vec{x_i}}{\vec{\alpha}_i}\changed{;p_{\alpha,\const_i}(\vec{x}_i)}}{\Psi_i}{e_i}{\beta}{\Phi}
 }{
   \tjudge{\Gamma, \typed{x}{\alpha}}{\Psi}{
     {\begin{array}[t]{l}
       \flstk{match }x\flstk{ with }
       \flst{| }\const_1\vec{x}_1\flst{ -> }e_1
       \flst{| }\dots
       \flst{ | }\const_n \vec{x}_n\flst{ -> }e_n
     \end{array} } }{\beta}{\Phi}
}

    \end{equation*}
    Wlog. we assume that $x\sigma = \const_i \vec{v}$. We then conclude the case by applying the first premise of rule, followed by an application of the MIH:
    \begin{align*}
      p_{\alpha,\const_i}(\vec{x}_i) \Rightarrow  \Psi((\Gamma,\typed{x}{\alpha})\sigma) &= \Psi(\Gamma\sigma,\const_i \vec{x}_i)\\
                           &= \Psi_i((\Gamma,\typed{\vec{x}_i}{\vec{\alpha}_i})\sigma)& \text{Premise}\\
                           &\ge c + \Phi(v) & \text{MIH}\\
    \end{align*}
  \item[Case.] Assume $\Pi$ ends as follows
    \begin{equation*}
      \inferrule*[lab=\rlabel{\erule{True}}]{
\eval{\sigma}{c}{e_1}{\flstk{true}} \\
\eval{\sigma}{d}{e_2}{v}}
{\eval{\sigma}{c + d}{\flstk{if }e_1\flstk{ then }e_2\flstk{ else }e_3}{v}}

    \end{equation*}
    We may assume that $\Xi$ ends in an application of the rule $\ruleite$,
    \begin{equation*}
      \inferrule*[right=\rlabel{\ruleite}]{
  \changed{\tjudge{\Gamma}{\Psi_1}{e_1}{\Bool}{0}} \\
  \tjudge{\Gamma\changed{; \wfguard{e_1}}}{\Psi_2}{e_2}{\alpha}{\Phi} \\
  \tjudge{\Gamma\changed{; \wfguard{\neg e_1}}}{\Psi_2}{e_3}{\alpha}{\Phi}
}{
  \tjudge{\Gamma}{\Psi_1 + \Psi_2}{\flstk{if }e_1\flstk{ then }e_2\flstk{ else }e_3}{\alpha}{\Phi}
}

    \end{equation*}
    First we apply the MIH for $\tjudge{\Gamma}{\Psi_1}{e_1}{\Bool}{\Phi_{\downarrow}}$. Then since $e_1$ evaluates to $\true$, the predicate $e_1$ holds, which allows us to apply the MIH a second time for $\tjudge{\Gamma; e_1}{\Psi_2}{e_2}{\alpha}{\Phi}$ 
    \begin{align*}
      (\Psi_1 + \Psi_2)(\Gamma\sigma) &= \Psi_1(\Gamma\sigma) + \Psi_2(\Gamma\sigma) \\
                      &\geq \Psi_2(\Gamma\sigma)\\
                      &\geq c +\Phi(v) & \text{MIH}
    \end{align*}
    The symmetric case with $\eval{\sigma}{c}{e_1}{\flstk{false}}$, follows from the same argument.
    
    \item[Case.] Consider $e = \flstk{let }x\flst{ = }e_1\flstk{ in }e_2$, that is, $\Pi$ ends in the following rule:
  \begin{equation*}
    \inferrule*[lab=\rlabel{\erule{Let}}]{
        \eval{\sigma}{c_1}{e_1}{w}\\
        \eval{\sigma[x \mapsto w]}{c_2}{e_2}{v}       
}{
        \eval{\sigma}{c_1 + c_2}{\flstk{let }x\flst{ = }e_1\flstk{ in }e_2}{v}
}

  \end{equation*}
  where $c = c_1 + c_2$. In this case $\Xi$ ends in an application of the $\rulelet$ rule.
  \begin{equation*}
\inferrule*[right=\rlabel{\rulelet}]{
  \Psi(\Gamma,\Delta) = \Psi_1(\Gamma) + \Psi_{\Delta}(\Delta) + \Psi_{\text{mix}}(\Gamma,\Delta)\\
  \tjudge{\Gamma}{\Psi_1}{e_1}{\alpha}{\Phi_1} \\\\
  \Psi_2(\typed{x}{\alpha},\Delta) = \Phi_1(e_1) + \Psi_{\Delta}(\Delta) + \Omega(\typed{x}{\alpha},\Delta) \\
  \tjudge{\Delta, \typed{x}{\alpha}}{\Psi_2}{e_2}{\beta}{\Phi}\\\\
  \bigwedge \Psi^i(\Gamma) \geq \Phi^i(\typed{x}{\alpha}) \Rightarrow  \Psi_{\text{mix}}(\Gamma,\Delta) \geq \Omega(\typed{x}{\alpha},\Delta)\\
  \tjudgenacf{\Gamma{\mid}\Psi^i}{e_1}{\alpha}{\Phi^i} \\
  }{
  \tjudge{\Gamma, \Delta}{\Psi}{\flstk{let }x\flst{ = }e_1\flstk{ in }e_2}{\beta}{\Phi}
}

  \end{equation*}
  
  By MIH we have
  \begin{equation*}
    \Psi_2((\typed{x}{\alpha}, \Delta)\sigma[x \mapsto w]) \geq c_2 + \Phi(v) \tpkt
  \end{equation*}
  So in order to prove the theorem, i.e.
  \begin{equation*}
    \Psi((\Gamma, \Delta)\sigma) \geq c_1 + c_2 +  \pot{v} \tkom
  \end{equation*}
  we are left to show
  \begin{equation*}
    \Psi(\Gamma, \Delta)\sigma \geq c_1 + \Psi_2((\typed{x}{\alpha}, \Delta)\sigma[x \mapsto w]) \tpkt
  \end{equation*}
  Instatiating the MIH as $\Psi_1(\Gamma) \geq c_1 + \Phi_1(w)$ and applying it to the first premise of the rule we get:
  \begin{align*}
     \Psi((\Gamma,\Delta)\sigma) &\geq c_1 + \Phi_1(w) + \Psi_{\Delta}(\Delta\sigma) + \Psi_{\text{mix}}((\Gamma,\Delta)\sigma)
  \end{align*}
  Applying the MIH to the cost-free typings in the premise, we obtain the key inequality
  \begin{equation*}
    \Psi_{\text{mix}}(\Gamma,\Delta) \geq \Omega(\typed{x}{\alpha},\Delta) \tkom
  \end{equation*}
  which we instantiate, as
  \begin{equation*}
    \Psi_{\text{mix}}((\Gamma,\Delta)\sigma) \geq \Omega((\typed{x}{\alpha},\Delta)\sigma[x \mapsto w]) \tpkt
  \end{equation*}
  Applying this inequality to previous one concludes the case. 
\item[Case.] Next we assume that $\Xi$ ends in an application of the rule $\ruleapp$ and $\Pi$ with $\eval{\sigma}{c}{f x_1\, \dots\, x_k}{v}$. We then have $f(y_1,\dots, y_k) = e \in P$ and since $P$ is well-typed, we have
  \begin{mathpar}
    \tjudge{\typed{\vec{y}}{\vec{\alpha}}}{\Psi}{e}{\beta}{\Phi} \and \text{and} \and \tjudge{\typed{\vec{y}}{\vec{\alpha}}}{\Psi_\text{cf}}{e}{\beta}{\Phi_\text{cf}}
  \end{mathpar}
  Applying the MIH with respect to $\Pi' = \eval{\sigma'}{c}{e}{v}$, we get $\Psi((\typed{\vec{y}}{\vec{\alpha}})\sigma') \geq c + \Phi(v)$, and $\Psi_\text{cf}((\typed{\vec{y}}{\vec{\alpha}})\sigma') \geq \Phi_\text{cf}(v)$. Let $k \in \Qplus$. We then obtain the theorem as follows.
  \begin{align*}
    (\Psi + k \cdot \Psi_\text{cf})((\typed{\vec{x}}{\vec{\alpha}})\sigma)
    &= \Psi((\typed{\vec{x}}{\vec{\alpha}})\sigma) + k \cdot \Psi_\text{cf}((\typed{\vec{x}}{\vec{\alpha}})\sigma)\\
    &= \Psi((\typed{\vec{y}}{\vec{\alpha}})\sigma') + k \cdot \Psi_\text{cf}((\typed{\vec{y}}{\vec{\alpha}})\sigma')\\
    &\geq c + \Phi(v)+ k \cdot \Phi_\text{cf}(v)\\
    &= c + \Phi(v) +  k \cdot \Phi_\text{cf}(v) \\
    &= c + (\Phi +  k \cdot \Phi_\text{cf})(v)
  \end{align*}
\item[Case.] Finally consider $\Xi$ ending in an application of the rule $\ruletick$ and let $\Pi$ derive $\eval{\sigma}{c + a}{\tick[a]{e}}{v}$. We can then apply the MIH with respect to $\eval{\sigma}{c}{e}{v}$ and the premise of the rule and obtain the theorem as follows.
  \begin{align*}
    \Psi(\Gamma\sigma) & \ge c + \Phi(v) & \text{MIH}\\
          & = c + \Phi(v) - a + a & \\
          & = c + a + \Phi(v) - a &\\
          & = c + a + (\Phi - a)(v)&
  \end{align*}
  \end{description}
\end{proof}

\section{Implementation}
\label{sec:implementation}
In this section, we briefly sketch the implementation of our prototype. At its core is a type inference algorithm for the algorithmic type system introduced above. The tool is written in Haskell and follows a modular design, enabling the analysis of different data types and potential functions. Its workflow can be divided into two main stages:
\begin{enumerate}
\item \emph{Frontend}: Parses the input program into an abstract syntax tree (AST), performs type inference, applies program transformations, and contextualises the program.
\item \emph{Resource Analysis}: Derives types with resource annotations, generates the corresponding constraint system, solves it optimally, and computes amortised bounds.
\end{enumerate}

\subsubsection{Frontend.} The abstract grammar from Figure~\ref{fig:syntax} is implemented using the parser combinator library \texttt{megaparsec}\footnote{\url{https://github.com/mrkkrp/megaparsec}}. After parsing, we apply a variant of Hindley–Milner type inference to assign plain types to every expression. These types are stored in the AST and later guide the inference of resource templates.

The program is then transformed into let-normal form, as required by the semantics and  type system. Finally the AST is supplemented with context clues that guide tactic generation for the main type derivation algorithm (see the discussion of heuristics below).

\subsubsection{Resource Analysis.} To derive type annotations, we apply Algorithm~\ref{alg:inference}. The subroutine $(*)$ below, which implements the algorithmic type system described in Section~\ref{sec:inference}, constitutes the core of this algorithm. It generates constraints by applying the rules of the type system according to a combination of syntax-directed and heuristically generated tactics. In line with the main soundness theorem, there exists a cost-free variant of this routine, which applies the corresponding cost-free versions of the type rules, when deriving a cost-free signature.

\begin{algorithm}[H]
  \caption{Algorithm for Amortised Resource Analysis}
  \label{alg:inference}
  \begin{algorithmic}[1]
    \State \textbf{Input:} program $\mathcal{P}$, analysis mode
    \State \textbf{Output:} proof tree and coefficients or proof tree and unsat core
    \State add signature constraints based on the analysis mode
    \State create a fresh resource template representing a potential function $\Phi_\alpha$ for type $\alpha$.
    \For{each function $f$ in program $\mathcal{P}$}
      \State create a fresh template $\Psi$ as LHS for $f$'s signature, 
      \State \hspace{1.2cm} and set RHS to $\Phi_\alpha$ or 0.
      \State create fresh cost-free signatures for $f$.
      \For{each signature of $f$}
        \State generate constraints and proof tree via syntax-directed type derivation (*)
      \EndFor
    \EndFor
    \State $\text{coefficients} \gets \text{solve and optimize the constraints}$
    \If{the constraints are satisfiable}
      \State \Return proof tree, coefficients, and compute bounds
    \Else
      \State \Return proof tree and unsat core
    \EndIf
  \end{algorithmic}
\end{algorithm}

We allow the annotation of each function with a mode pragma, where default is to generate one signature with costs where the RHS of the signature is set the potential function for the corresponding data type. Additionally we support two alternative modes. With \lstinline|{-# MODE worst_case #-}|, the potential function is set to zero, thereby performing a worst-case analysis and \lstinline|{-# MODE hybrid #-}|, creates two costed signatures for the function --- one with potential and one worst-case signature. This special mode is required to perform the amortised analysis of \texttt{meld} for rank-biased leftist heaps. Here the rule $\ruleapp$ combines the results of a worst-case analysis with a cost-free typing to obtain lowest amortised cost.

When it comes to cost-free signatures the default behaviour is to generate a single cost-free signature, which is sufficient for the majority of examples, but we allow to set the number of cost-free signatures per function with another pragma, e.g.~\lstinline|{-# NUM_CF_SIGS 2 #-}|. For the analysis of skew heaps with the sums-of-logs potential, for example, two cost-free signatures are required for both \texttt{meld} and \texttt{bal}.

\subsubsection{Heuristics} As mentioned earlier the type inference algorithm (*) needs to apply the structural rules in syntax directed fashion. A straightforward solution would be to apply the structural rules before every syntax-directed rule, but this leads to an intractable number of constraints~\cite{leutgeb2021atlas}. Instead, we apply them heuristically only when needed and restrict expert knowledge to parametrisations of $\ctrule{W}$, thereby reducing the number of generated constraints. To support this approach, the AST is annotated with contextual cues that trigger applications of structural rules, after which inference proceeds syntax-directed.

For example, the cue \texttt{pseudo leaf} marks an expression that corresponds to a leaf in the resulting proof tree but is not bound by a let expression. Encountering this cue triggers an application of $\ctrule{W}$ with monotonicity expert knowledge. Similarly, $\ruleshift$ is applied to ticked function applications, and $\rulewvar$ is applied before $\ctrule{Const}$ and $\ctrule{Var}$ to eliminate superfluous variables from the context. Most of these heuristics originate from Leutgeb et al.~\cite{leutgeb2021atlas}; we adapted them where necessary for the new benchmarks.

\paragraph{Template Languages.} Different template languages are implemented as Haskell modules and can be selected by the user on a file-by-file basis. For example, the pragma \lstinline|{-# POTENTIAL (Tree Base: logr) #-}| specifies that the data type \texttt{Tree Base} is to be analysed using the template language \texttt{logr}, which corresponds to $\tlangSol$ instantiated with $\params = (0,\sfrac{1}{2},0)$. This information is used by the inference algorithm to generate the appropriate resource templates.

\paragraph{SMT Interface.}
During type inference we collect template constraints represented by a custom data type, that is subsequently translated into a solver-API specific representation. In addition to those obligations, we add signature constraints, e.g. we ensure potential functions are positive and when performing type checking we constrain the signatures to the given values. Currently our implementation only supports the popular solver Z3~\cite{de2008z3}, which provides an option for optimisation, a crucial feature for our inference approach. In addition to the actual result, we produce a proof tree in HTML format. This facilitates debugging, as we can track the constraints identified by Z3’s unsat core and trace them back to their origin in the type derivation.

\paragraph{Optimisation.}
To derive the optimal resource types, we make use of Z3’s optimisation feature, guided by a cost function that models the amortised costs. For this cost-function we express the amortised costs, $\Psi(\vec{x}) - \Phi(\vec{x})$, where $\vec{x}$ denotes the function arguments, purely symbolically, i.e.~in terms of LRA coefficients. We define the symbolic cost $\symcost(\Psi, \Phi)$ of a signature, in this example for the $\tlangSol$, as 
\begin{align*}
  &\symcost(\Psi, \Phi)_x = q_{x} - p_{\val}\\
  &\symcost(\Psi, \Phi)_{(ax + c)} = q_{(ax + c)} - p_{(a\val + c)}\\
  &\symcost(\Psi, \Phi)_{i} = q_{i} - p_{i}
\end{align*}
where $\coeff(\Psi) = [q_i]$ and $\coeff(\Phi) = [p_i]$. To apply this approach to multi-argument potential functions, we follow~\cite{sleator1986self} and set $\phi(x_1, \dots, x_n) := \phi(x_1) + \dots + \phi(x_n)$. This is required, for example, when determining the costs of \texttt{meld\ x\ y}. Even-though we only ever derive $\Phi(\val)$, a single argument function, we can easily define $\Phi(\typed{x}{\Tree},\typed{y}{\Tree})$ as $\Phi(\typed{x}{\Tree}) + \Phi(\typed{y}{\Tree})$.
We guide optimisation by introducing a weight function that penalises larger terms:
\begin{align*}
  &w((c)) = 1\\
  &w(x) = 100\\
  &w((x)) = 1\\
  &w((\vec{a}\vec{x} + c)) = (1 + 1^T\vec{a} + 2 * c)^2
\end{align*}
This leads us to the definition of the cost function $c(\Psi, \Phi)$ as the weighted sum of the symbolic costs for every term in the resource templates of a given signature.
\begin{equation*}
  c(\Psi, \Phi) = \sum_i w(i) \cdot \symcost(\Psi, \Phi)_i 
\end{equation*}
\begin{example}
Recall that for \texttt{swap} we have amortised costs $\Psi(\typed{x}{\Tree}) - \Phi(\typed{x}{\Tree}) = \log\abs{x} + \phi(x) - \phi(x) = \log\abs{x}$. This solution has the following costs:
\begin{align*}
  c(\Psi, \Phi) & = w((x)) \cdot \symcost(\Psi, \Phi)_{\log\abs{x}} + w(x) \cdot \symcost(\Psi, \Phi)_{\phi(x)}\\
          & = 4 \cdot \symcost(\Psi, \Phi)_{\log\abs{x}}\\
            & = 4
\end{align*}
\end{example}

\section{Value Variables}
\label{ap:value-variables}

In the following we describe a technique that allows to include references to the result expression in the cost terms, e.g. $\log(\abs{\texttt{meld}\ x\ y})$. With our type system it is not possible to unfold such terms with their definition immediately, as one would do in a manual proof. Instead we extend the typing context by value variables such that for any typing judgement $\tjudge{\Gamma, \val}{\Psi}{e}{\alpha}{\Phi}$, we define $\val := e$. By doing so we allow the analysis to treat references to the expression purely symbolically until unfolding becomes trivial, i.e.~when $\val$ refers to a value. In principle, we could conclude the derivation with a rule analogous to $\rulevar$. Instead, we introduce the rule $\ruleunfold$, which leaves the derivation open and thereby enables further applications of structural rules—most notably the rule $\rulew$.
\begin{equation*}
  
\end{equation*}
To demonstrate that $\Psi'(\vec{x}) = \Psi(\vec{x}, \typed{\val}{\alpha})$ can be discharged in an algorithmic fashion, we give the following intuition. We can split up $\Psi(\vec{x}, \typed{\val}{\alpha})$ based on the occurring variables:
\begin{equation*}
  \Psi(\vec{x}, \typed{\val}{\alpha}) = \Psi_{\vec{x}}(\vec{x}) + \Psi_{\val}(\val)
\end{equation*}
Of course this only works if there are no terms in $\Psi$ that mix variables from $\vec{x}$ and $\val$, a requirement that we enforce in our resource templates. We can than for a resource function $\Psi_1$ with an application of $\ruleconst$ verify that $\Psi_1(\vec{x}) = \Psi_{\val}(\val)$. Finally we need to ensure $\Psi'(\vec{x}) = \Psi_{\vec{x}} + \Psi_1(\vec{x})$ through suitable constraints that model coefficient wise addition. 
\begin{example}
  \label{ex:unfold}
  The type derivation of \texttt{meld} with value variables requires the following application of $\ruleunfold$, in which the analysis defers reasoning about $\log\abs{\val}$ until its value is known.
  \begin{equation*}
    \inferrule*[lab=\scriptsize\ruleunfold]{
      \inferrule*[right=\scriptsize\ruleconst]{ }{
          \tjudgeml{\typed{t}{\Tree}, \typed{z}{\Tree} \mid \log(\abs{z} + \abs{t})}{\tree{z}{a}{t}}{\Tree}{\log\abs{\val}}
      }
      \inferrule*[right=\scriptsize\rulew]{
        \inferrule*[right=\scriptsize\ruleconst]{ }{
          \tjudgeml{\typed{t}{\Tree}, \typed{z}{\Tree} \mid \phiPw(z) + [\abs{z} < \abs{t}] \\ + \phiPw(t)}{\tree{z}{a}{t}}{\Tree}{\phiPw(\val)}
        }
      }{
          \tjudgeml{\typed{t}{\Tree}, \typed{z}{\Tree} \mid \log(\abs{z} + \abs{t}) - \log\abs{z} \\+ \phiPw(t) + \phiPw(z)}{\tree{z}{a}{t}}{\Tree}{\phiPw(\val)}
      } 
    }{
      \tjudgeml{\typed{\val}{\Tree}, \typed{t}{\Tree}, \typed{z}{\Tree} \mid \log\abs{\val} - \log\abs{z} + \phiPw(t) + \phiPw(z)}{\tree{z}{a}{t}}{\Tree}{\phiPw(\val)}
    }    
  \end{equation*}

\end{example}

\subsection{Optimisation for Value Variables}

With the introduction of value variables, the optimization approach introduced in Section~\ref{sec:implementation}, encounters a fundamental issue: multiple terms may represent the same value. In other words, distinct terms can be equivalent and cancel each other out in the cost function. This leads to an unbounded objective, which makes optimisation infeasible. For instance, in the case of skew heaps we have the identity $\log(\abs{\val} + 1) = \log(\abs{x} + \abs{y})$ so the solver could freely assign arbitrarily large positive or negative coefficients to these terms without affecting correctness.

To address this issue, we make use of a standard trick from optimisation: instead of allowing equivalent terms to offset each other directly, we introduce absolute values into the cost function. Since SMT solvers cannot handle non-linear absolute values natively, we encode them in a linear way using epigraph variables.

Concretely, for each symbolic cost term $t$ we introduce a fresh variable $u_t \geq 0$ that represents $\abs{t}$. This is enforced by the pair of linear constraints $u_t \geq t$, $u_t \geq -t$.
Thus $u_t$ always bounds the absolute value of $t$. In the objective function we then minimise (or weight) $u_t$ in addition to $t$. This retains the original cost semantics while ensuring the optimisation remains bounded.

\subsection{Experimental Comparison}

As mentioned earlier, using value variables yields simpler proofs. In particular, for $\texttt{meld}$ no cost-free typings are required, which makes the derivation computationally less expensive. Table~\ref{tab:compare-value} compares inference times for skew heaps under our two variants of potential functions, as well as for rank-biased heaps. The results show a significant difference in inference time, which correlates with the number of cost-free typings. This highlights a trade-off: with value variables, the additional reasoning steps captured by cost-free typings are shifted outside the analysis. Interestingly, the number of constraints alone does not determine performance. 

\begin{table}[h]
  \centering
  \caption{Comparison of type inference for \texttt{meld} with and without value variables. Inference time is the median of five runs. Here \#cf denotes the number of required cost-free signatures and \#constraints the number of generated constraints.}
  \label{tab:compare-value}
  \begin{tabularx}{\textwidth}{Xlcc@{\hspace{8ex}}lcc}
    \toprule
    Benchmark & \multicolumn{3}{c}{Plain} & \multicolumn{3}{c}{Value Variables} \\
              & Time & \#cf & \#constraints & Time & \#cf & \#constraints\\
    \midrule
    piecewise & 32s & 1 & 9895 & 14s & 0 & 18461 \\
    sum-of-logs & 12m13s & 2 & 8729 & 1m49s  & 0 & 14152\\
    rank-biased & 1h08m43s & 1& 11399 & 21s & 0 & 10470\\
    \bottomrule
  \end{tabularx}
\end{table}

\section{Ordering of Logarithmic Terms respecting Guard Predicates}
\label{sec:guard-mono}

To understand our approach of incorporating both monotonicity and guard predicates into expert knowledge instantiation, we take a close look at how logarithmic terms are compared. First, we define a partial order over the set of logarithmic template terms of a particular resource template $Q$. Let the variables $\vec{x}$ represent tree sizes, so we have $\vec{x} \geq \vec{1}$.
To define an order relation $\le$ between two terms $\log(\vec{a} \cdot \vec{x} + c)$ and $\log(\vec{b}\cdot \vec{x} + d)$, we must produce a certificate for
\[
(\vec{a} - \vec{b})^T\vec{x} \leq d - c
\]
holding for all feasible $\vec{x}$. 

If the only constraints on $\vec{x}$ are $\vec{x} \ge \vec{1}$, the expression $(\vec{a} - \vec{b})^T\vec{x}$ is bounded above iff all coefficients are non-positive, and the maximum is attained at $\vec{x} = \vec{1}$. This yields the simple condition:
\begin{equation*}
  \log(\vec{a} \cdot \vec{x} + c) \leq \log(\vec{b}\cdot \vec{x} + d) \quad\Leftrightarrow\quad \vec{a} - \vec{b} \leq \vec{0}, \quad (\vec{a} - \vec{b})^T\vec{1} \leq d - c.
\end{equation*}

However, in our new setting we also have additional linear constraints $(C-D)\vec{x} \le \vec{0}$ arising from guard predicates involving tree sizes. In this case, the above simplification is no longer sufficient. A complete $\vec{x}$-free certificate can be derived from linear programming duality:

\begin{lemma}
Let $\vec{p} := \vec{a} - \vec{b}$ and $\beta := d - c$. Then 
\begin{equation}
  \label{eq:logMonoDual1}
  \vec{p}^T \vec{x} \le \beta \quad \forall\, \vec{x} \ge \vec{1},\ (C-D)\vec{x} \le 0  
\end{equation}
holds if and only if there exist $\vec{v} \ge 0$ such that
\begin{equation}
  \label{eq:logMonoDual3}
  (C-D)^T \vec{v} \geq \vec{p}, \quad \vec{p}^T\mathbf{1} - \vec{v}^T(C-D)\mathbf{1} \le \beta.  
\end{equation}
\end{lemma}
\begin{proof}
Put $\vec{y}=\vec{x}-\vec{1}$. Then the result is an instantiation of Farkas' Lemma (Lemma~\ref{lem:farkas}), where $A:=C-D$, $\vec{x}:=\vec{y}$, $\vec{b}:=-(C-D)\vec{1}$, $\vec{u}:=\vec{p}$, $\lambda:=\beta-\vec{p}^T \vec{1}$ , and $\vec{f}:=\vec{v}$.
\end{proof}
The simple rule above is recovered as the special case $\vec{v} = \vec{0}$ when no extra constraints are present. We note that in our actual implementation, the certificate is restricted to $\vec{v} = \vec{1}\lambda$ with $\lambda \in \{0,1\}$, to simplify computation, which still yields all required inequalities in our context.


\fi
\end{document}